\documentclass[12pt, draftclsnofoot, onecolumn]{IEEEtran}
\usepackage{listings}

\usepackage{multicol}

\usepackage{blindtext}
\usepackage{amsfonts}
\usepackage{amsmath, amsthm, amssymb}
\usepackage{times}
\usepackage{url}
\usepackage{verbatim}
\usepackage{array}
\usepackage{mathrsfs}
\usepackage{enumitem}
\usepackage{textcomp}

\newcommand{\subparagraph}{}
\usepackage[compact]{titlesec}
\usepackage[ruled, lined, linesnumbered]{algorithm2e}
\SetAlFnt{\small}

\SetCommentSty{mycommfont}

\usepackage{amsfonts}
\usepackage{bbm}
\usepackage{soul}
\usepackage{color}
\usepackage{graphics}
\usepackage{graphicx,epsfig}
\usepackage{lipsum,graphicx,subcaption}
\graphicspath{{figures/}}
\usepackage{caption}
\usepackage{subcaption}

\usepackage{cleveref}
\usepackage{wrapfig}
\usepackage{url}
\usepackage{url}
\usepackage{amsmath,amssymb,epsfig,amsthm,psfrag,epsf, multirow,wrapfig, graphicx}
\usepackage[english]{babel}
\usepackage{epstopdf}
\usepackage{float}
\usepackage{color}
\usepackage[table,xcdraw]{xcolor}
\usepackage{amsmath}
\usepackage{amsfonts}
\usepackage{amssymb}
\usepackage{mathtools}
\DeclarePairedDelimiter{\ceil}{\lceil}{\rceil}
\usepackage{tabularx}
\usepackage{adjustbox}
\usepackage{physics}

\usepackage{amsfonts}
\usepackage{amssymb}
\usepackage{caption}
\usepackage{subcaption}
\usepackage{tabularx}
\usepackage{comment}
\usepackage[mathscr]{euscript}
\usepackage{amsbsy}
\usepackage{stmaryrd}
\usepackage{mathtools}
\usepackage[utf8]{inputenc}
\graphicspath{{figs/}}

\usepackage{booktabs}
\usepackage{graphicx,epsfig}
\usepackage{graphics}
\usepackage{multirow}
\usepackage{rotating}
\usepackage{caption}
\usepackage{subcaption}
\usepackage{amsmath}
\usepackage{amssymb}
\usepackage{makecell}

\DeclareMathOperator*{\argmax}{arg\,max}
\DeclareMathOperator*{\argmin}{arg\,min}
\usepackage{mathtools}
\usepackage{comment}
\usepackage{multirow}

\usepackage{footnote}
\usepackage{tablefootnote}
\usepackage{scalerel}

\usepackage{tabularx}
\usepackage{mdframed}
\usepackage{lipsum}
\makeatletter
\newcommand*{\rom}[1]{\expandafter\@slowromancap\romannumeral #1@}
\makeatother
\usepackage{nccmath}

\usepackage{relsize}

\newtheorem{theorem}{Theorem}

\newtheorem{corollary}{Corollary}

\newtheorem{lemma}{Lemma}

\renewcommand\IEEEkeywordsname{Keywords}

\newcommand{\note}[1]{{\color{blue}{#1}}} 
\newcommand{\inred}[1]{{\color{red}{#1}}} 
\newcommand{\ingreen}[1]{{\color{green}{#1}}} 

\DeclareMathSymbol{\shortminus}{\mathbin}{AMSa}{"39}

\newcommand{\smallDelta}{{\text{$\scaleto{\Delta}{5pt}$}}}

\allowdisplaybreaks
\begin{document}
\def\eg{\mbox{\em e.g.}, }

\title{Joint Spatio-Temporal Precoding for Practical Non-Stationary Wireless Channels}

\author{
\small 
\begin{tabular}[t]{c@{\extracolsep{5em}}c} 
Zhibin Zou, Maqsood Careem, Aveek Dutta & Ngwe Thawdar \\
Department of Electrical and Computer Engineering & US Air Force Research Laboratory \\ 
University at Albany SUNY, Albany, NY 12222 USA & Rome, NY, USA \\
\{{zzou2, mabdulcareem, adutta\}@albany.edu} & 
ngwe.thawdar@us.af.mil
\end{tabular}
}


    
\maketitle

\begin{abstract}

The high mobility, density and multi-path evident in modern wireless systems makes the channel highly non-stationary.
This causes temporal variation in the channel distribution that leads to the existence of time-varying joint interference across multiple degrees of freedom (DoF, \eg users, antennas, frequency and symbols), which renders conventional precoding sub-optimal in practice. 
In this work, we derive a High-Order Generalization of Mercer’s Theorem (HOGMT), which decomposes the multi-user non-stationary channel into two (dual) sets of jointly orthogonal subchannels (eigenfunctions), that result in the other set when one set is transmitted through the channel. 
This duality and joint orthogonality of eigenfuntions ensure transmission over independently flat-fading subchannels.
Consequently, transmitting these eigenfunctions with optimally derived coefficients eventually mitigates any interference across its degrees of freedoms and forms the foundation of the proposed joint spatio-temporal precoding. 
The transferred dual eigenfuntions and coefficients directly reconstruct the data symbols at the receiver upon demodulation, thereby significantly reducing its computational burden, by alleviating the need for any complementary post-coding. 
Additionally, the eigenfunctions decomposed from the time-frequency delay-Doppler channel kernel are paramount to extracting the second-order channel statistics, and therefore completely characterize the underlying channel. 
We evaluate this using a realistic non-stationary channel framework built in Matlab and show that our precoding achieves ${\geqslant}$4 orders of reduction in BER at SNR${\geqslant}15$dB 
in OFDM systems for higher-order modulations and less complexity compared to the state-of-the-art precoding.

\end{abstract}
\section{Introduction}



Precoding has been widely investigated for stationary channels, 
where the orthogonality along each DoF 
is enforced by decomposing them using linear algebraic tools (\eg singular value decomposition (SVD) or QR decomposition \cite{Cho2010MIMObook}) 
leading to capacity achieving strategies \cite{fatema2017massive,CostaDPC1983} typically under the block fading assumption. 
The statistical non-stationarity that is evident in modern and next Generation propagation environments including V2X, mmWave, and massive-MIMO channels, 
\cite{wang2018survey,huang2020general,mecklenbrauker2011vehicular}, leads to \textit{catastrophic} error rates even with state-of-the-art precoding \cite{AliNS0219} 
as such capacity-achieving strategies optimized for stationary channels do not ensure interference-free communication when the channel distribution changes over time.  
The non-stationarity in such channels engender joint interference across multiple dimensions (space (users/ antennas), time-frequency or delay-Doppler) in communication systems that leverage multiple degrees of freedom (\eg MU-MIMO, OFDM, OTFS \cite{OTFS_2018_Paper}). This time-varying joint interference renders conventional decomposition techniques incapable of achieving flat-fading.
Our solution to the above addresses a challenging open problem in the literature \cite{2006MatzOP}: ``how to decompose non-stationary channels into independently fading sub-channels (along each degree of freedom) and how to precode using them", which is central to both characterizing channels and minimizing interference. 
The multi-user non-stationary channel is represented as a 4-dimensional space (user) time-varying impulse response and is acquired from the CSI obtained from each receiver (user). 
The core of our precoding is the decomposition of this asymmetric channel, by generalizing Mercer's Theorem \cite{1909Mercer} to high-dimensional asymmetric processes, into 2-dimensional 
eigenfunctions that are jointly orthogonal across the DoF. These eigenfunctions serve as independently flat-fading subchannels, and precoding using them is the key to canceling the time-varying joint interference that exists across the DoF. 
The second order statistics of non-stationary channels vary across time-frequency and delay-Doppler (4-dimensions) and therefore, such channels 
can be represented as atomic channels with 4-dimensional asymmetric coefficients \cite{MATZ20111}. Unlike recent literature that only partially characterize the non-stationary channel using a select few local statistics~\cite{2018PMNS, bian2021general}, the 2-dimensional eigenfunctions decomposed from 4-dimensional asymmetric coefficients are used to extract any second-order statistics of the non-stationary channels that completely characterizes its distribution. Since any wireless channel model (\eg deterministic, stationary, frequency flat or selective) can be extracted from the general non-stationary channel kernel, the extracted eigenfunctions lead to a unified method to characterize the statistics of any wireless channel.

\begin{figure}[t]
    \centering
    \includegraphics[width=\linewidth]{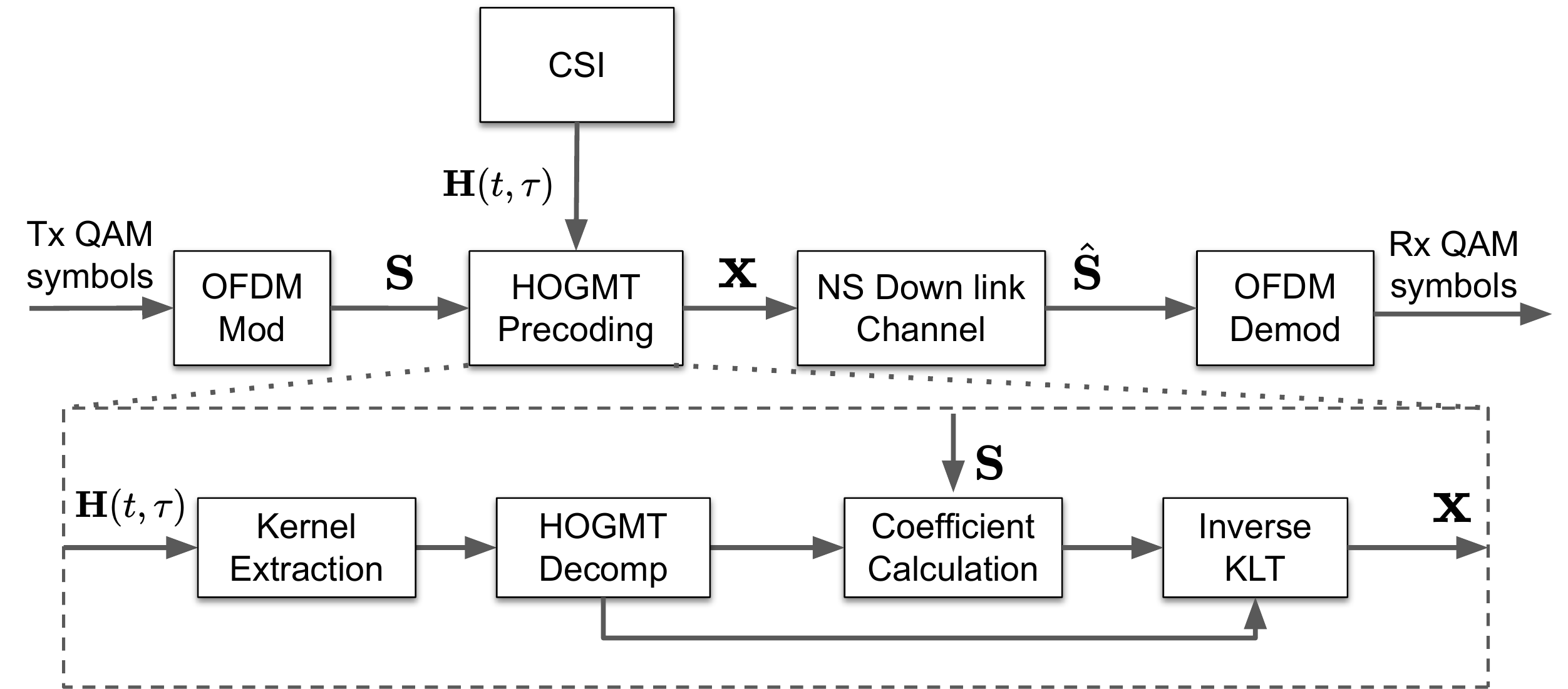}
    \caption{System view of HOGMT-based precoding for non-stationary channel}
    \label{fig:system}
\end{figure}

Figure \ref{fig:system} shows the system view of joint spatio-temporal precoding.
The spatio-temporal CSI is used to extract a 4-dimensional kernel, which can be decomposed into dual 2-dimensional (space-time) eigenfunctions by derived High Order Generalized Mercer's Theorem (HOGMT).
Thus decorrelate the space-time domain at the transmitter and the receiver.  
The spatio-temporal eigenfunctions corresponding to the receiver are used to derive optimal coefficients that minimize the least square error in the transmitted and received symbols. Then combining their dual spatio-temporal eigenfunctions with these coefficients by inverse KLT.
Since the eigenfunctions are independently and jointly orthonomal sub-channels over space and time, 
precoding using them warrants flat-fading (interference-free communication) even in the presence of joint space-time interference. 
Further, these transmitted (precoded) symbols directly reconstruct the data symbols at the receiver when combined with calculated coefficients. 
Therefore, unlike existing precoding methods that require complementary decoding at the receiver\cite{Cho2010MIMObook}, we alleviate any need for complex receiver processing thereby significantly reducing its computational burden. 
Finally, the precoded symbols are scheduled to each user and are processed through the conventional transmitter signal processing blocks (\eg CP/ guard insertion) before transmission. To our best of knowledge, precoding for NS channel is first proposed in our previous work~\cite{Zou2022NP}, which is extended in this work.
Additionally, we investigate the computational complexity of our precoding and show that it achieves lower complexity than Dirty Paper Coding (Section \ref{sec:complexity}). 
Therefore, we make the following contributions: 

\noindent
\textbf{1) Decomposing non-stationary channels into flat-fading subchannels:} The 4-dimensional channel kernels is decomposed into dual jointly orthogonal subchannels (eigenfunctions) that are flat-fading in the eigen-domain (Section \ref{sec:NSdec}).\\
\textbf{2) A Unified Characterization of Wireless Channels:} The eigenfunctions decomposed from 4-dimensional coefficients of atomic channels completely characterizes the non-stationary channel, which can generalize to any wireless channel (Section \ref{sec:stat}).\\
\textbf{3) Joint Spatio-Temporal Precoding:} The eigenfunctions decomposed from 4-dimensional channel kernels are leveraged to cancel the joint spatio-temporal interference in non-stationary channels (Section \ref{sec:precoding}).\\
\textbf{4) Post-coding free Precoding:} The precoded symbols transmitting through the channel directly
reconstruct the modulated sybols at receivers without complementary step, alleviating additional computational
burden (Section \ref{sec:precoding}).

\section{Background \& Related work}
\label{App:related_NLP}
\begin{table*}
\setlength{\textfloatsep}{0.1cm}
\setlength{\tabcolsep}{0.2em}
    \centering
    \scriptsize
    \def\arraystretch{1}
    \begin{tabular}{|p{3.5cm}|p{3cm}|p{5cm}|p{4.2cm}|}
\hline
\multicolumn{1}{|c|}{\textbf{Applications}} &
  \multicolumn{1}{c|}{\textbf{Source of NS}} &
  \multicolumn{1}{c|}{\textbf{Degree of NS}} &
  \multicolumn{1}{c|}{\textbf{SoTA BER}} 
 \\ \hline
\textbf{Vehicle-to-Everything (V2X)} \newline 
\textbf{High-Speed Train (HST)} \newline 
\textbf{Unmanned Aerial Vehicle (UAV)} 
  &
  High Tx-Rx mobility, weather,  time-varying scatterers, AoA, AoD and Doppler, UAV attitude. 
  & V2X: SI${=}$10–40 m at 5–15 km/h \cite{Renaudin2010NonStationaryNM,Renaudin2009C2C} \newline
  HST: SI${=}$1.37-2.84m at 2.6GHz and 18MHz BW \cite{Zhou2018HST} \newline
  UAV: SI${=}$15m 
  at C band with 50 MHz BW \cite{matolak2012air,khuwaja2018survey}
  
  & V2X: BER${>}0.01$ for BPSK at SNR 20dB \cite{Jaime2020V2X}
  
\\\hline

\textbf{MIMO}\newline 
\textbf{Massive MIMO}\newline
\textbf{Extra-large MIMO (XL-MIMO)}\newline
\textbf{MU-MIMO}
  &
  Time varying multipath, delay spread and spatial visibility regions. 
  & MIMO: 0-3m Correlation Matrix Distance \cite{herdin2005correlation}\newline
  Massive/XL/MU-MIMO: AoA shift from 100$^o{-}$80$^o$ 
  at 2.6 GHz with 50 MHz BW \cite{payami2012channel}
  & MIMO: BER${>}0.07$ for QPSK at SNR 20dB \cite{herdin2005correlation} with 8${\times}$8 MIMO
  Massive/XL/MU-MIMO: BER${>}10^{-2}$ with 320 antenna elements 
  \cite{malkowsky2017world}
  
\\\hline

\textbf{mmWave}, \textbf{THz} \newline
\textbf{Underwater}, \textbf{Satellite-Com} \newline
\textbf{Visible Light Communications} \newline
\textbf{Reconf Intelligent Surface} \newline
\textbf{Hypersonic Reentry Com (HRC)}
&
  Time varying blockage, reflection angles and NLoS, atmospheric absorption topology and opportunistic access.   
&
mmWave: Delay spread${<}$20 ns at 28–73 GHz \cite{wang2018survey,Lovnes1994mmW}\newline
HRC: Coherence time ${\approx}5\mu s$ \cite{shi2019effective}
&
THz: BER${>}0.02$ for 16-QAM at 26m distance and frequency of 350 GHz \cite{Wang2020268mTW}\newline
Underwater: BER${>}0.003$ QPSK at SNR 20dB \cite{sharif2000closed}\newline
HRC: BER${>}10^{-1}$ \cite{shi2019effective}
  \\ \hline
\end{tabular}
\caption{Overview of applications in non-stationary channels: Impairments, SoTA and BER performance.}
\label{tab:ns}
\end{table*}
We categorize the related work into three categories:

\noindent
\textbf{Non-stationarity of Wireless Channels:}
Table~\ref{tab:ns} provides ample evidence in the literature for the existence of non-stationarity (NS) in modern wireless channels. 
Non-stationarity is primarily attributed to temporal and sometimes spatial variation in the transceivers and the dynamic nature of the scattering environment, and is measured in terms of the stationarity interval (SI) \cite{Renaudin2010NonStationaryNM} in time or space. 
Depending on the features of the channel non-stationarity may arise from the time-varying Doppler in V2X, HST and UAV channels, from the time-varying multipath in MIMO channels and its variants and due to time varying blockage in mmWave, THz and VLC channels. 
Unfortunately, state of the art (SoTA) methods applied to NS channels are only able to achieve a modest error rate that is inadequate to support high data rate wireless applications like mobile AR/VR/XR, aerial communications and 4K/8K HDR video streaming services. These will require joint spatio-temporal precoding for NS channels that is capable of achieving orders of magnitude improvement in Bit Error Rate (BER) across all types of channels and applications.

\noindent
\textbf{Precoding in Non-Stationary Channels:}
Although precoding non-stationary channels is unprecedented in the literature \cite{AliNS0219}, we list the most related literature for completeness. 
The challenge in precoding non-stationary channels is the time-dependence of statistics and the channel cannot be modeled as the time-independent matrix. This leads to suboptimal performance using state-of-the-art precoding techniques like Dirty Paper Coding (DPC). Though DPC is theoretically interference-free with perfect CSI, the current implementation by QR decomposition targeting for separate channel matrices, unable to capture the variation over time.
Meanwhile, recent literature present attempt to deal with imperfect CSI by modeling the error in the CSI \cite{HatakawaNLP2012, HasegawaTHP2018, GuoTHP2020, DietrichTHP2007, CastanheiraPGS2013, WangTHP2012, MazroueiTDVP2016, Jacobsson1DAC2017}, they are limited by the assumption the channel or error statistics are stationary or WSSUS at best. 

\noindent
\textbf{Spatio-Temporal Precoding:}
While, precoding has garnered significant research, spatio-temporal interference is typically treated as two separate problems, where spatial precoding at the transmitter aims to cancel inter-user and inter-antenna interference, while equalization at the receiver mitigates inter-carrier and inter-symbol interference.
Alternately, \cite{OTFS_2018_Paper} proposes to modulate the symbols such that it reduces the cross-symbol interference in the delay-Doppler domain, but requires equalization at the receiver to completely cancel such interference in practical systems. 
Moreover, this approach cannot completely minimize the joint spatio-temporal interference that occurs in non-stationary channels since their statistics depend on the time-frequency domain in addition to the delay-Dopper domain (explained in Section \ref{sec:pre}).
While spatio-temporal block coding techniques are studied in the literature \cite{Cho2010MIMObook} they add redundancy and hence incur a communication overhead to mitigate interference, which we avoid by precoding. 
These techniques are capable of independently canceling the interference in each domain, however are incapable of mitigating interference that occurs in the joint spatio-temporal domain in non-stationary channels. 
We design a joint spatio-temporal precoding that leverages the extracted 2-D eigenfunctions from non-stationary channels to mitigate interference that occurs on the joint space-time dimensions, which to the best of our knowledge is unprecedented in the literature. 

\section{Models \& Preliminaries}
\label{sec:pre}
%
%

\subsection{Non-stationary wireless channel model}
\label{sec:pre}

The wireless channel is typically expressed by a linear operator $H$, and the received signal $r(t)$ is given by $r(t){=}Hs(t)$, where $s(t)$ is the transmitted signal. The physics of the impact of $H$ on $s(t)$ is described using the delays and Doppler shift in the multipath propagation~\cite{MATZ20111} given by \eqref{eq:H_delay_Doppler},
\begin{equation}
    r(t) = \sum\nolimits_{p=1}^P h_p s(t-\tau_p) e^{j2\pi \nu_p t}
    \label{eq:H_delay_Doppler}
\end{equation}
where $h_p$, $\tau_p$ and $\nu_p$ are the path attenuation factor, time delay and Doppler shift for path $p$, respectively. 
\eqref{eq:H_delay_Doppler} is expressed in terms of the overall delay $\tau$ and Doppler shift $\nu$ \cite{MATZ20111} in \eqref{eq:S_H}, 
%
\begin{align}
    &r(t) = \iint S_H(\tau, \nu) s(t{-} \tau) e^{j2\pi \nu t} ~d\tau ~d\nu \label{eq:S_H} \\
    & = \int L_H(t, f) S(f) e^{j2\pi tf} ~df 
     = \int h(t, \tau) s(t{-} \tau) ~d\tau \label{eq:h_relation}
\end{align}
where $S_H(\tau, \nu)$ is the \textit{(delay-Doppler) spreading function} of channel $H$, which describes the combined attenuation factor for all paths in the delay-Doppler domain. $S(f)$ is the Fourier transform of $s(t)$ and the time-frequency (TF) domain representation of $H$ is characterized by its \textit{TF transfer function}, $L_H(t,f)$, which is obtained by the 2-D Fourier transform of $S_H(\tau, \nu)$ as in \eqref{eq:TF_SH}. 
The time-varying impulse response $h(t,\tau)$ is obtained as the Inverse Fourier transform of $S_H(\tau, \nu)$ from the Doppler domain to the time domain as in  \eqref{eq:h_SH}.
\begin{align}
    &L_H(t,f) {=} \iint S_H(\tau, \nu) e^{j2\pi (t\nu{-} f \tau)} ~d\tau ~d\nu \label{eq:TF_SH}\\
    &h(t,\tau) {=} \int S_H(\tau, \nu) e^{j2\pi t \nu} ~d\nu \label{eq:h_SH}
\end{align}
Figure~\ref{fig:example_Ht} and \ref{fig:example_Lh} show the time-varying impulse response and TF transfer function for a NS channel, respectively.
%
\begin{figure*}[t]
\begin{subfigure}{.265\textwidth}
  \centering
\includegraphics[width=1\linewidth]{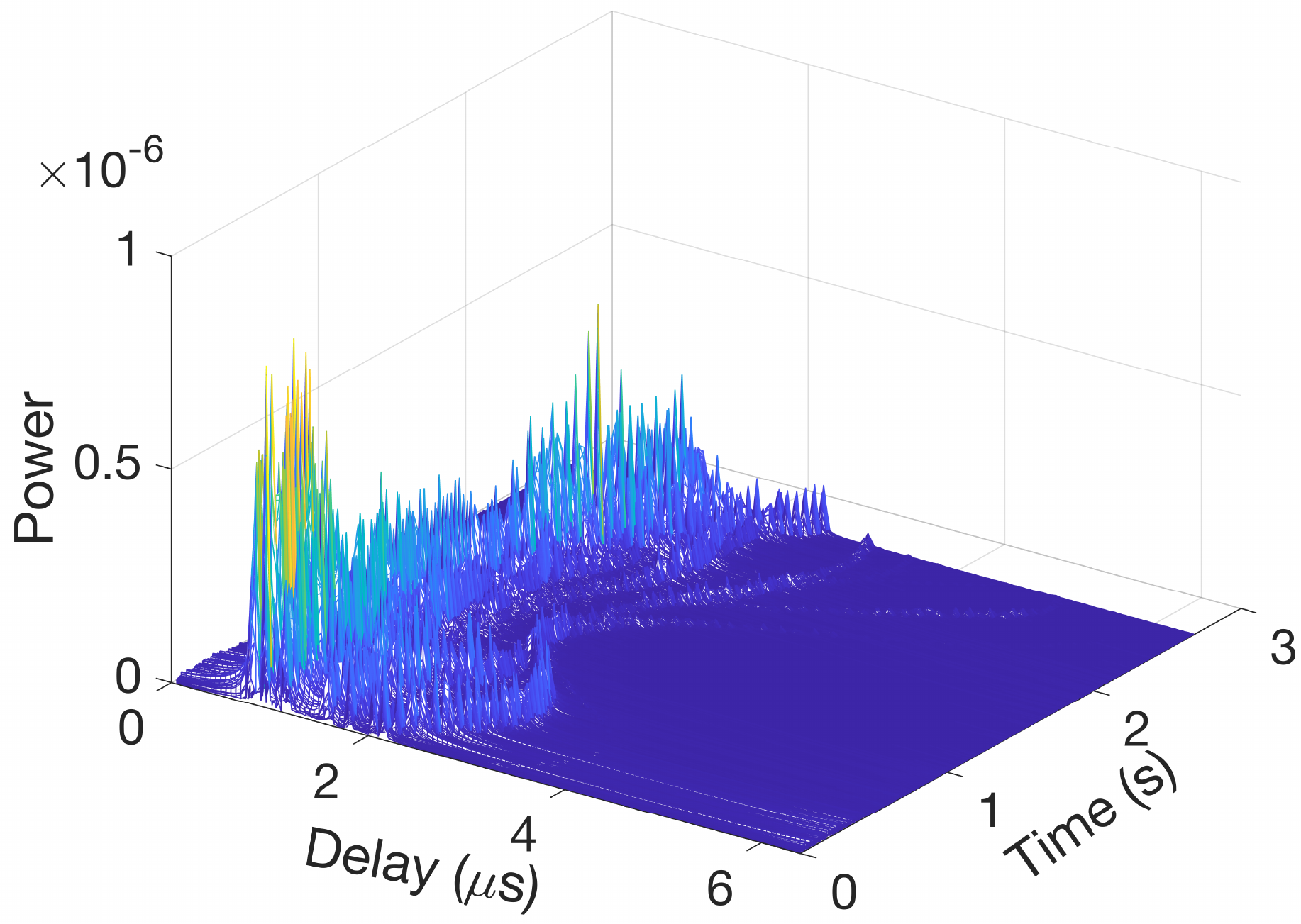}
  \caption{The time-varying impulse response of antenna 1 to user 1}
  \label{fig:example_Ht}
\end{subfigure}
\qquad
\quad
\begin{subfigure}{.265\textwidth}
  \centering
  \includegraphics[width=1\linewidth]{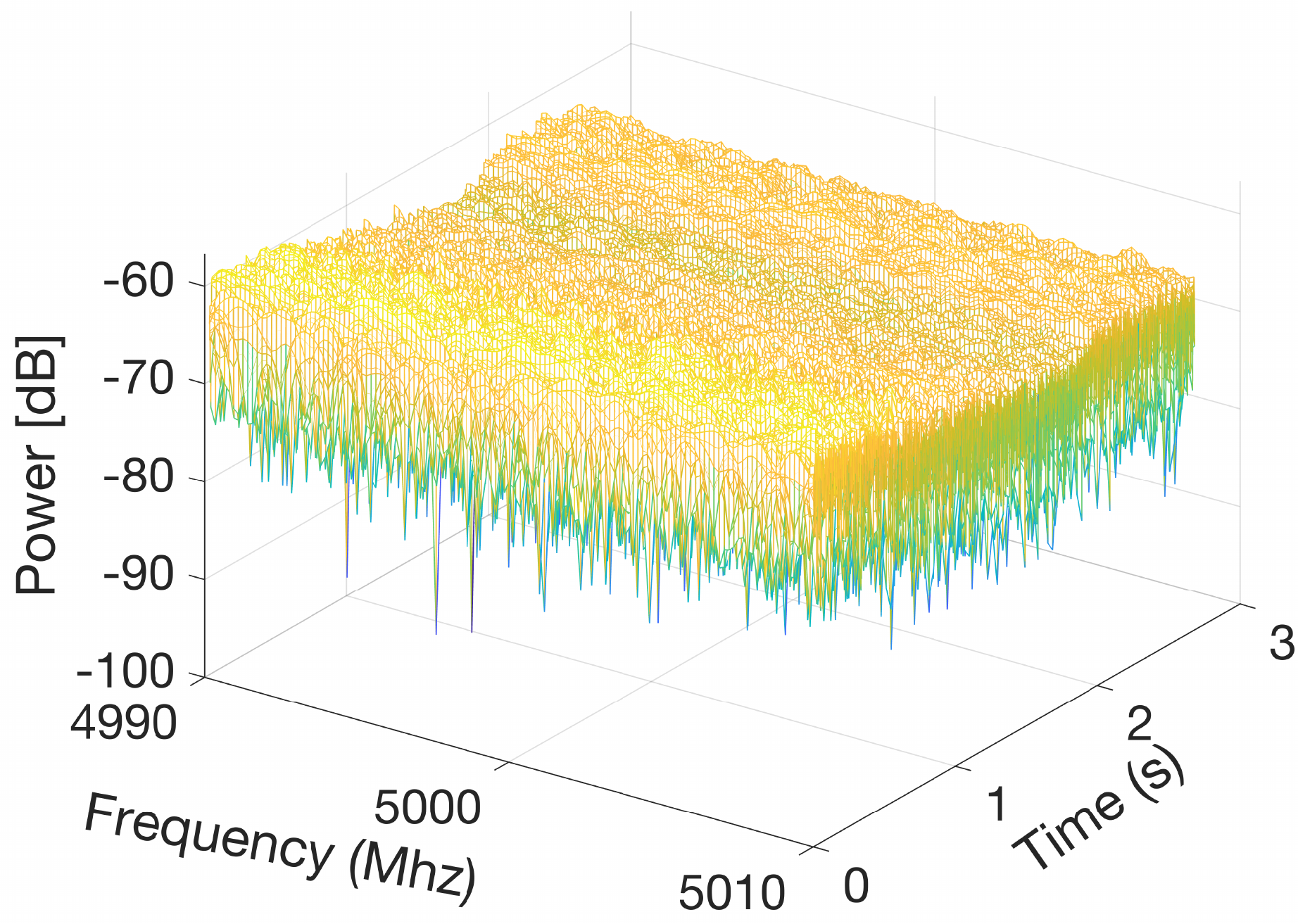}
  \caption{The time-frequency transfer function of antenna 1 to user 1}
  \label{fig:example_Lh}
\end{subfigure}
\qquad
\quad
\begin{subfigure}{.325\textwidth}
  \centering
  \includegraphics[width=1\linewidth]{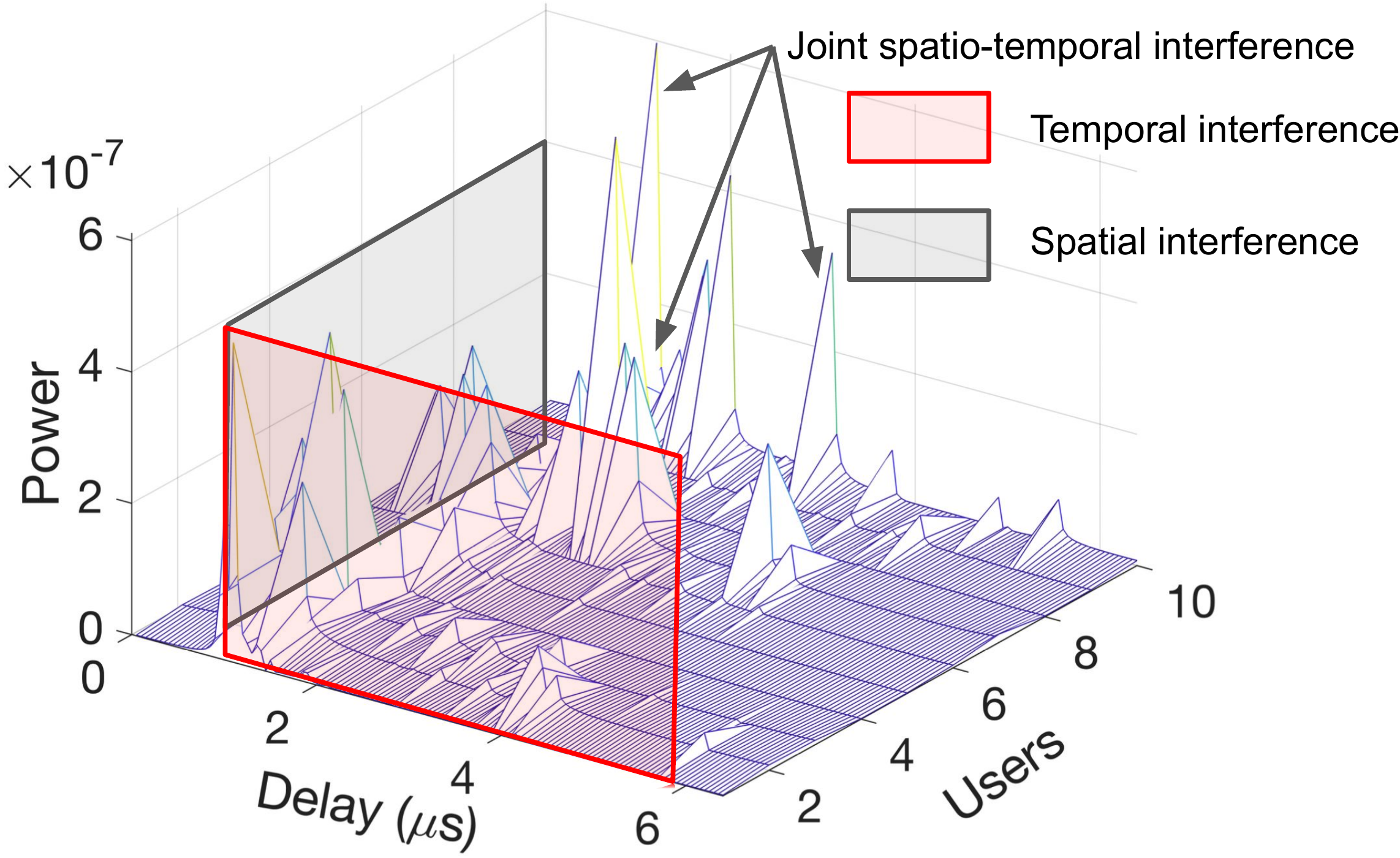}
  \caption{The spatio-temporal channel gains of antenna 1 at time 1}
  \label{fig:joint_interference}
\end{subfigure}
\caption{Illustration of a non-stationary channel}
\label{fig:toychannel}
\end{figure*}

\subsection{Statistics of non-stationary channels}

For stationary channels, the TF transfer function is a stationary process and and the spreading function is a white process (uncorrelated scattering) with 
\begin{align}
    &\mathbb{E}\{L_H(t, f) L_H^*(t',f')\} {=} R_H(t{-}t',f{-}f')\\
    &\mathbb{E}\{S_H(\tau, \nu) S_H^*(\tau',\nu')\} {=} C_H(\tau,\nu) \delta (\tau{-} \tau') \delta(\nu{-} \nu')
\end{align}

where $\delta(\cdot)$ is the Dirac delta function.
$C_H(\tau,\nu)$ and $R_H(t-t',f-f')$ are the \textit{scattering function} and \textit{TF correlation function}, respectively, which are related via 2-D Fourier transform, 

\begin{equation}
    C_H(\tau,\nu) = \iint R_H(\Delta t, \Delta f) e^{-j2\pi(\nu \Delta t -  \tau \Delta f)} ~d\Delta t ~d\Delta f
\end{equation}
In contrast, for non-stationary channels, the TF transfer function is non-stationary process and the spreading function is a non-white process. 
Therefore, a \textit{local scattering function} (LSF) $\mathcal{C}_H(t,f;\tau,\nu)$ \cite{MATZ20111} is defined to extend $C_H(\tau,\nu)$ to the non-stationary channels in \eqref{eq:LSF}. 
Similarly, the \textit{channel correlation function} (CCF) $\mathcal{R}(\Delta t, \Delta f;\Delta \tau, \Delta \nu)$ generalizes $R_H(\Delta t, \Delta f)$ to the non-stationary case in \eqref{eq:CCF}.
\begin{align}
\begin{split}
    &\mathcal{C}_H(t,f ; \tau,\nu)  \label{eq:LSF}\\
    & {=} {\iint} R_L(t, f; \Delta t, \Delta f) e^{-j2\pi(\nu \Delta t -  \tau \Delta f)} ~d\Delta t ~d\Delta f  \\
    & {=} \iint R_S(\tau, \nu; \Delta \tau, \Delta \nu) e^{-j2\pi(t \Delta \nu -  f \Delta \tau)} ~d\Delta \tau ~d\Delta \nu \\
\end{split}\\
\begin{split}
    &\mathcal{R}(\Delta t, \Delta f;\Delta \tau, \Delta \nu) \label{eq:CCF}\\
    & {=} \iint R_L(t, f; \Delta t, \Delta f) e^{-j2\pi(\Delta \nu  t -  \Delta \tau  f)} ~d t ~d f \\
    & {=} \iint R_S(\tau, \nu; \Delta \tau, \Delta \nu) e^{-j2\pi(\Delta t \nu -  \Delta f  \tau)} ~d \tau ~d \nu
\end{split}
\end{align}
where 
\begin{align}
    &R_L(t, f; \Delta t, \Delta f){=}\mathbb{E}\{L_H(t, f {+} \Delta f) L_H^*(t {-} \Delta t, f)\} \nonumber\\
    &R_S(\tau, \nu; \Delta \tau, \Delta \nu){=}\mathbb{E}\{S_H(\tau, \nu {+} \Delta \nu) S_H^*(\tau{-} \Delta \tau, \nu)\} \nonumber
\end{align}

For stationary channels, CCF reduces to TF correlation function as
\begin{align}
    &\mathcal{R}(\Delta t, \Delta f;\Delta \tau, \Delta \nu){=}R_H(\Delta t, \Delta f) \delta(\Delta t) \delta(\Delta f) \nonumber
\end{align}

\subsection{Multi-user non-stationary channel model}
For precoding, 
we express the spatio-temporal downlink channel response 
by extending the time-varying response $h(t,\tau)$ to incorporate multiple users. Without loss of generality, for convenience of exposition let us consider the case where each user has a single antenna, i.e., MISO case. Denotes $h_{u,u'}(t,\tau)$ \cite{almers2007survey} as the time-varying impulse response between the $u'$\textsuperscript{th} transmit antenna and the $u$\textsuperscript{th} user (For MIMO case, $h_{u,u'}(t,\tau)$ is a matrix). 
Thus the received signal in~(3) is extended to
the multi-user case of $h(t,\tau)$ and is given by~(11),
\begin{align}
\label{eq:H_t_tau}
\mathbf{H}(t,\tau) = \begin{bmatrix}
 h_{1,1}(t,\tau) &\cdots & h_{1,u'}(t,\tau)  \\
 \vdots& \ddots &  \\
h_{u,1}(t,\tau) &  & h_{u,u'}(t,\tau)
\end{bmatrix}
\tag{11}
\end{align}
%
Therefore, the received signal in \eqref{eq:h_relation} is extended as in \eqref{eq:r_mu_discrete}.
\begin{table*}[b]
\begin{align}
    & r_u (t) {=} \int \sum_{u'} h_{u,u'}(t,\tau) s_{u'}(t-\tau) d\tau + v_{u}(t) \label{eq:r_mu_discrete} \\
    &{=} \underbrace{h_{u,u}(t,0)s_{u}(t)}_\text{Signal with the attenuation coefficient} + \underbrace{\sum_{u'\neq u} h_{u,u'}(t,0)s_{u'}(t)}_\text{Spatial interference} + \underbrace{\int h_{u,u}(t,\tau) s_{u}(t-\tau) d\tau}_\text{Temporal interference} + \underbrace{\int \sum_{u' \neq u} h_{u,u'}(t,\tau) s_{u'}(t-\tau)d\tau}_\text{Joint spatio-temporal interference} + \underbrace{v_{u}(t)}_\text{Noise} \label{eq:r_mu}
\end{align}
\end{table*}
The first term in \eqref{eq:r_mu} is desired signal with fading effects. 
Spatial and temporal interference correspond to the second and third terms in \eqref{eq:r_mu} and the black and red regions in figure~\ref{fig:joint_interference}, respectively. 
Additionally, the space-time kernel induces joint spatio-temporal due to the interference from delayed symbols from other users as shown in the figure and the last term of \eqref{eq:r_mu}. 
Further, since the spatio-temporal signals in practice are 2-dimensional ($r_{u}(t)$ or $s_{u}(t)$), canceling all the above interference, necessitates a method to decompose the asymmetric 4-D channel in \eqref{eq:H_t_tau} into 2-D 
independently fading subchannels.

The CSI processing including estimating the 4-D spatio-temporal channel at the receiver (to obtain CSI), compressing and transmitting, and predicting CSI by outdated CSI are widely investigated in the literature \cite{2020CESrivastava,2005CEXiaoli,2008CEMilojevic, Guo2022CSI, Maqsood2020Prediction} and hence is not the focus of this work. Therefore, in the following sections we investigate the ability to cancel all interference for non-stationary channels under perfect CSI.

Therefore, in the following sections we investigate the ability to cancel all interference (precode) non-stationary channels under perfect CSI. 
\section{Non-stationary channel decomposition}
\label{sec:NSdec}

4-D channel decomposition into orthonormal 2-D subchannels is unprecedented the literature, but is essential to mitigate joint interference in the 2-D space and to completely characterize non-stationary channels (Any channel can be generated as a special case of the non-stationary channel. 
Therefore a precoding for non-stationary channels would generalize to any other wireless channel \cite{MATZ20111}).  
While SVD is only capable of decomposing LTI channels, Karhunen–Loève transform (KLT) \cite{wang2008karhunen}
provides a method to decompose random process into component eigenfunctions of the same dimension. 
However, 
KLT is unable to decompose the multi-user time-varying 4-D channel in \eqref{eq:H_t_tau}, into orthonormal 2-D space-time eigenfunctions (DoF decorrelation), and therefore cannot mitigate interference on the joint space-time dimensions. 
Mercer's theorem provides a method to decompose symmetric 2-D kernels into the same eigenfunctions along different dimensions, however, it cannot directly decompose 4-D channel kernels due to their high-dimensionality and since the channel kernel in \eqref{eq:H_t_tau} is not necessarily symmetric in the 4 dimensions. 
Therefore, we derive a generalized version of Mercer's Theorem for asymmetric kernels and extend it to higher-order kernels, which decomposes the asymmetric 4-D channel into 2-D jointly orthogonal subchannels. Consequently, this leads to flat-fading communication and mitigates joint spatio-temporal interference.
The above decomposition techniques are compared in Table~\ref{tab:decomp}, where only SVD has been used in the literature for precoding.


\begin{table*}[h]
\setlength{\textfloatsep}{0.1cm}
\setlength{\tabcolsep}{0.2em}
\caption{Comparison of HOGMT with other channel decomposition methods}
\renewcommand*{\arraystretch}{1.2}
\centering
\begin{tabular}{|l|c|c|c|c|c|}
\hline
\textbf{Decomposition methods} & \textbf{Time-varying} & \textbf{Asymetric} & \textbf{DoF decorrelation} & \textbf{High-order} \\ \hline
SVD: $H = U \Sigma V^*$             &   & \checkmark & \checkmark &          \\ \hline
Karhunen–Loève Transform (KLT): $X(t) = \Sigma_n \sigma_n \phi_n(t)$             &  \checkmark & \checkmark &  &                        \\ \hline
Mercer's Theorem: $K(t,t') = \Sigma_n \lambda_n \phi_n(t) \phi_n(t')$            &  \checkmark & & \checkmark &                    \\ \hline
Generalized Mercer's Theorem (GMT): $K(t,t') = \Sigma_n \sigma_n \psi_n(t) \phi_n(t')$&  \checkmark & \checkmark & \checkmark &  \\ \hline
\rowcolor[HTML]{FFFFBF}HOGMT: $K(\zeta_1{,}{...}{,}\zeta_P{;} \gamma_1{,}{...}{,} \gamma_Q) {=} {\sum_{n}} \sigma_n \psi_n(\zeta_1{,}{...}{,}\zeta_P) \phi_n(\gamma_1{,}{...}{,}\gamma_Q)$ 
    &  \checkmark & \checkmark & \checkmark & \checkmark
\\ \hline
\end{tabular}
\label{tab:decomp}
\end{table*}
\begin{lemma}
\label{lemma:GMT}
(Generalized Mercer's theorem (GMT))
The decomposition of a 2-dimensional process $K {\in} L^2(X {\times} Y)$, where $X$ and $Y$ are square-integrable zero-mean 
processes, 
is given by,
\begin{align}
    K(t, t') = \sum\nolimits_{n=1}^{\infty} \sigma_n \psi_n(t) \phi_n(t')
    \label{eq:GMT}
\end{align}
where $\sigma_n$ is a random variable with $\mathbb{E}\{\sigma_n \sigma_{n'}\} {=} \lambda_n \delta_{nn'}$, and $\lambda_n$ is the $n$\textsuperscript{th} eigenvalue. $\psi_n(t)$ and $\phi_n(t')$ are eigenfunctions.
\end{lemma}
\begin{proof}
Consider a 2-dimensional process $K(t,t') \in L^2(Y \times X)$, where $Y(t)$ and $X(t')$ are square-integrable zero-mean random processes with covariance function $K_{Y}$ and $K_{X}$, respectly. 
The projection of $K(t, t')$ onto $X(t')$ is obtained as in \eqref{eq:projection},
\begin{align}
\label{eq:projection}
    & C(t) = \int K(t, t') X(t') ~dt'
\end{align}

Using \textit{Karhunen–Loève Transform} (KLT), $X(t')$ and $C(t)$ are both decomposed as in \eqref{eq:X_t_C_t}, 
\begin{align}
\label{eq:X_t_C_t}
\begin{aligned}
X(t') = \sum_{i = 1}^{\infty} x_{i} \phi_{i}(t') \quad \text{and,} \quad
C(t) = \sum_{j = 1}^{\infty} c_{j} \psi_{j}(t) 
\end{aligned}
\end{align}

where $x_i$ and $c_j$ are both random variables with $\mathbb{E}\{x_i x_{i'}\} {=} \lambda_{x_i} \delta_{ii'}$ and $\mathbb{E}\{c_j c_{j'}\} {=} \lambda_{c_j} \delta_{jj'}$.  $\{\lambda_{x_i}\}$, $\{\lambda_{x_j}\}$ $\{\phi_i(t')\}$ and $\{\psi_j(t)\}$ are eigenvalues and eigenfuncions, respectively.
%
%
Let us denote $n{=}i{=}j$ and $\sigma_n {=} \frac{c_n}{x_n}$, and assume that $K(t,t')$ can be expressed as in \eqref{eq:thm_K_t},
\begin{align}
\label{eq:thm_K_t}
    K(t,t') = \sum_n^\infty \sigma_n \psi_{n}(t) \phi_{n}(t')
\end{align}
We show that \eqref{eq:thm_K_t} is a correct representation of $K(t,t')$ by proving \eqref{eq:projection} holds under this definition. 
We observe that by substituting \eqref{eq:X_t_C_t} and \eqref{eq:thm_K_t} into the right hand side of \eqref{eq:projection} we have that,
\begin{align}
    & \int K(t, t') X(t') ~dt' 
    {=} \int \sum_n^\infty \sigma_n \psi_{n}(t) \phi_{n}(t') \sum_{n}^{\infty} x_{n} \phi_{n}(t') ~dt' \nonumber \\
    & {=} {\int} {\sum_n^\infty} \sigma_n x_n \psi_n(t) |\phi_n(t')|^2 {+} {\sum_{n'\neq n}^ \infty} \sigma_{n} x_{n'} \psi_{n}(t) \phi_{n}(t') \phi_{n'}^*(t') d t' \nonumber \\
    & = \sum_n^\infty c_n \psi_n(t) = C(t)
\end{align}
which is equal to the left hand side of \eqref{eq:projection}. 
Therefore, \eqref{eq:thm_K_t} is a correct representation of $K(t,t')$.
\end{proof}

%
%
\noindent
%
From Lemma~\ref{lemma:GMT}, by letting $\rho(t, t'){=}\psi_n(t) \phi_n(t')$ in \eqref{eq:GMT} we have \eqref{eq:GMT_2D}, 
\begin{equation}
    K(t, t') {=} \sum_{n=1}^{\infty} \sigma_n \rho_n(t, t')
    \label{eq:GMT_2D}
\end{equation}
where the 2-D kernel is decomposed into random variable ${\sigma_n}$ with constituent 
2-D eigenfunctions, $\rho(t, t')$, this serves as an extension of KLT 
to 2-D kernels.
A similar extension leads to the derivation of KLT for N-dimensional kernels which is 
key to deriving Theorem \ref{thm:hogmt}.



\noindent
\fbox{\begin{minipage}{0.97\linewidth}
\begin{theorem}
\label{thm:hogmt}
(High Order GMT (HOGMT)) The decomposition of $M {=} Q{+}P$ dimensional kernel $K {\in} L^M(X {\times} Y)$, where $X(\gamma_1,\cdots,\gamma_Q)$ and $Y(\zeta_1,\cdots,\zeta_P)$ are $Q$ and $P$ dimensional kernels respectively, that are square-integrable zero-mean random processes, is given by \eqref{eq:col},
\begin{align}
\label{eq:col}
K(\zeta_1{,}{...}{,}\zeta_P{;} \gamma_1{,}{...}{,} \gamma_Q) {=} {\sum_{n{=}1}^ \infty} \sigma_n \psi_n(\zeta_1{,}{...}{,}\zeta_P) \phi_n(\gamma_1{,}{...}{,}\gamma_Q)
\end{align}
where $\mathbb{E}\{\sigma_n \sigma_n'\} {=} \lambda_n \delta_{nn'}$. $\lambda_n$ is the $n$\textsuperscript{th} eigenvalue and $\psi_n(\zeta_1,\cdots,\zeta_P)$ and $\phi_n(\gamma_1,\cdots, \gamma_Q)$ are $P$ and $Q$ dimensional eigenfunctions respectively.
\end{theorem}
\end{minipage}}

\begin{proof}
Given a 2-D process $X(\gamma_1, \gamma_2)$, the eigen-decomposition using Lemma \ref{lemma:GMT} is given by,
\begin{equation}
\label{eq:thm1_1}
    X(\gamma_1, \gamma_2) = \sum_{n}^{\infty} x_{n} e_n(\gamma_1) s_n(\gamma_2)
\end{equation}

Letting $\psi_n(\gamma_1,\gamma_2) {=} e_n(\gamma_1) s_n(\gamma_2)$, and substituting it in \eqref{eq:thm1_1} we have that,

\begin{equation}
\label{eq:2d_klt}
    X(\gamma_1, \gamma_2) = \sum_{n}^{\infty} x_{n} \phi_n(\gamma_1,\gamma_2)
\end{equation}
where $\phi_n(\gamma_1,\gamma_2)$ are 2-D eigenfunctions with the property \eqref{eq:prop1}.
\begin{equation}
\label{eq:prop1}
\iint \phi_n(\gamma_1,\gamma_2) \phi_{n'}(\gamma_1,\gamma_2) ~d\gamma_1 ~d\gamma_2 = \delta_{nn'} 
\end{equation}

We observe that \eqref{eq:2d_klt} is the 2-D form of KLT. With iterations of the above steps, we obtain \textit{High-Order KLT} for $X(\gamma_1,\cdots,\gamma_Q)$ and $C(\zeta_1,\cdots,\zeta_P)$ as given by,
\begin{align}
   & X(\gamma_1,\cdots,\gamma_Q) = \sum_{n}^{\infty} x_{n} \phi_n(\gamma_1,\cdots,\gamma_Q) \\
   & C(\zeta_1,\cdots,\zeta_P) = \sum_{n}^{\infty} c_{n} \psi_n(\zeta_1,\cdots,\zeta_P)
\end{align}
where $C(\zeta_1,\cdots,\zeta_P)$ is the projection of $X(\gamma_1,\cdots,\gamma_Q)$ onto $K(\zeta_1,\cdots,\zeta_P; \gamma_1,\cdots, \gamma_Q)$.
Then following similar steps as in the proof of Lemma 2, we get \eqref{eq:col}. 
\end{proof}

Theorem~\ref{thm:hogmt} is applicable to any $M$ dimensional channel kernel. Examples of such channel kernels may include 1-D time-varying channels, 2-D time-frequency kernels for doubly dispersive channels \cite{2003MatzDP}, user, antenna dimensions in MU-MIMO channels and angles of arrivals and departures 
in mmWave channels.

\section{Characterization of Non-stationary channels}
\label{sec:stat}

Wireless channel characterization in the literature typically require several local and global (in space-time dimensions) higher order statistics to characterize or model non-stationary channels, due to their time-varying statistics. 
These statistics cannot completely characterize the non-stationary channel, however are useful in reporting certain properties that are required for the application of interest such as channel modeling, assessing the degree of stationarity etc.
Contrarily, we leverage the 2-dimensional eigenfunctions that are decomposed from the most generic representation of any wireless channel as a spatio-temporal channel kernel.
These spatio-temporal eigenfunctions can be used to extract any higher order statistics of the channel as demonstrated in Section \ref{sec:NSdec}, and hence serves as a complete characterization of the channel.
Furthermore, since this characterization can also generalize to stationary channels, it is a unified characterization for any wireless channel.
Beyond characterizing the channel, these eigenfunctions are the core of the precoding algorithm.
The analysis of non-stationary channels is 
complicated as its statistics vary across both time-frequency and delay-Doppler domains resulting in 4-D second order statistics \cite{Matz2005NS}, which motivates the need for a unified characterization of wireless channels\footnote{Any channel can be generated as a special case of the non-stationary channel. Therefore a characterization of non-stationary channels would generalize to any other wireless channel \cite{MATZ20111}.}.
Wireless channels are completely characterized by their statistics, however they are difficult to extract for non-stationary channels, due to their time dependence.
%
%
Therefore, we start by expressing the channel $H$ using an atomic channel $G$ and the 4-D channel kernel $\mathcal{H}(t, f ; \tau, \nu)$~\cite{Matz2005NS} as in \eqref{eq:atomic}, 
\begin{align}
\label{eq:atomic}
   H = \iiiint \mathcal{H}(t, f ; \tau, \nu) G_{t,f}^{\tau,\nu} ~dt ~df ~d\tau ~dv
\end{align}
where $G$ is a normalized $(||G||{=}1)$ linear prototype system whose transfer function $L_G(t, f)$ is smooth and localized about the origin of the TF plane. $G_{t,f}^{\tau,\nu} {=} S_{t, f{+}\nu} G S_{t{-}\tau, f}^+$ means that the atomic channel $G$ shifts the signal components localized at $(t{-}\tau, f)$ to $(t, f{+}\nu)$ on the TF plane. $S_{\tau, \nu}$ is TF shift operator defined as $(S_{\tau, \nu}s)(t) {=} s(t{-}\tau) e^{j2 \pi \nu t}$. 
Then the channel kernel $\mathcal{H}(t, f ; \tau, \nu)$ is given by \eqref{eq:compute_H}.
\begin{align}
\label{eq:compute_H}
&\mathcal{H}(t, f ; \tau, \nu)=\left\langle H, G_{t, f}^{\tau, \nu}\right\rangle  \\
& {=}\mathrm{e}^{j 2 \pi f \tau} \iint L_{H}\left(t^{\prime}, f^{\prime}\right) L_{G}^{*}\left(t^{\prime}{\shortminus}t, f^{\prime}{\shortminus}f\right) 
 \mathrm{e}^{{-}j 2 \pi\left(\nu t^{\prime}{-}\tau f^{\prime}\right)} \mathrm{d} t^{\prime} \mathrm{d} f^{\prime} \nonumber
 \end{align}


The statistics of any wireless channel 
can always be obtained from the above 4-D channel kernel. Therefore, decomposing this kernel into fundamental basis allows us to derive a unified form to characterize any wireless channel.
Theorem~\ref{thm:hogmt} ensures that the 4-D channel kernel in \eqref{eq:compute_H} is decomposed as in \eqref{eq:channel_kernel_decomp} into 2-D eigenfunctions that are jointly orthonormal in the time-frequency or delay-Doppler dimensions as in \eqref{eq:properties}. 
\begin{align}
\label{eq:channel_kernel_decomp}
&\mathcal{H}(t, f ; \tau, \nu) = \sum\nolimits_{n{=1}}^\infty \sigma_n \psi_n(t, f) \phi_n(\tau, \nu)\\
\label{eq:properties}
&\begin{aligned}
&& \iint \psi_n(t, f) \psi_{n'}^*(t, f) ~dt ~df {=} \delta_{nn'}  \\
&& \iint \phi_n(\tau, \nu) \phi_{n'}^*(\tau, \nu) ~d\tau ~d\nu {=} \delta_{nn'}
\end{aligned}
\end{align}
The variation across time-frequency delay-Doppler domains in the 4-D channel kernel is extracted by decomposing into separate 2-D eigenfuntions in time-frequency and delay-Doppler domains, respectively. The decorrelation of dimensionality and the orthonormal properties in \eqref{eq:properties} allow eigenfuntions and eigenvalues to extract statistics in either time-frequency and delay Doppler 4-D domains or separate 2-D domains, as shown in Corollary~\ref{col:characterization}.
\begin{corollary}
\label{col:characterization}
(Unified characterization for non-stationary channel by HOGMT) The statistics of the non-stationary channel is completely characterized by its eigenvalues and eigenfunctions 
obtained by the decomposition of $\mathcal{H}(t,f;\tau,\nu)$, which are summarized in Table~\ref{tab:cha}.
\setlength{\textfloatsep}{0.1cm}
\setlength{\tabcolsep}{0.2em}
\begin{table}[h]
\caption{Unified characterization of non-stationary channel}
\renewcommand*{\arraystretch}{1.2}
\centering
\begin{tabular}{|l|l|}
\hline
\textbf{Statistics} & \textbf{Eigen Characterization} \\ \hline
CCF $|\mathcal{R}(\Delta t, \Delta f;\Delta \tau, \Delta \nu)|$             &   ${\sum} \lambda_n |R_{\psi_n}(\Delta t{,}\Delta f)| |R_{\phi_n}(\Delta \tau{,}\Delta \nu)|$                     \\ \hline
LSF  $\mathcal{C}_H(t, f; \tau, \nu)$             &         ${\sum} \lambda_n |\psi_n( \tau,\nu)|^2 |\phi_n(t, f)|^2$               \\ \hline
Global scattering function   $\overline{C}_H(\tau{,} \nu)$               &        ${\sum} \lambda_n |\psi_n( \tau,\nu)|^2 $                \\ \hline
Local TF path gain $\rho_H^2 (t,f)$
    & ${\sum} \lambda_n |\phi_n( t, f)|^2 $  \\ \hline
Total transmission gain     $\mathcal{E}_H^2$                          &     ${\sum} \lambda_n$                   \\ \hline
\end{tabular}
\label{tab:cha}
\end{table}
\end{corollary}
\begin{proof}

Wireless channels are fully characterized by their (second order) statistics, which we calculate using the extracted eigenvalues and 2-D eigenfunctions.
The CCF is calculated as the correlations of $\mathcal{H}(t, f ; \tau, \nu)$ and is given by,
\begin{align}
    & \abs{ \mathcal{R}(\Delta t, \Delta f;\Delta \tau, \Delta \nu) } \label{eq:CCF3}\\
    &{=}  \Bigl| {\iiiint} \mathbb{E}\{\mathcal{H}^*(t {\shortminus} {\smallDelta} t{,} f {\shortminus} {\smallDelta} f{;} \tau {\shortminus} {\smallDelta} \tau{,} \nu {\shortminus} {\smallDelta} \nu) 
    \mathcal{H}(t{,}f{;}\tau{,} \nu)\} d t d f d \tau d \nu \Bigr|  \nonumber \\
    &{=} \sum\nolimits_{n{=1}}^\infty \lambda_n \abs{R_{\psi_n}(\Delta t,\Delta f)}  |R_{\phi_n}(\Delta \tau,\Delta \nu)| \label{eq:CCF2}
\end{align}
where \eqref{eq:CCF2} is obtained by substituting \eqref{eq:channel_kernel_decomp} in \eqref{eq:CCF3}. 
$R_{\psi_n}(\Delta t,\Delta f)$ and $R_{\phi_n}(\Delta \tau,\Delta \nu)$ are the correlations of $\psi_n(t, f)$ and $\phi_n(\tau, \nu)$, respectively. 
The LSF reveals the non-stationarities (in time or frequency) in a wireless channel and is given by the 4-D Fourier transform ($\mathbb{F}^{4}$) of the CCF as, 
%
%
\begin{align}
\label{eq:LSF_eigen}
    & \mathcal{C}_H(t, f; \tau, \nu) {=} \mathbb{F}^{4}\left\{\mathcal{R}({\smallDelta} t{,} {\smallDelta} f{;}{\smallDelta} \tau{,} {\smallDelta} \nu)\right\} \nonumber\\
    & {=} \iiiint \mathcal{R}({\smallDelta} t{,} {\smallDelta} f{;}{\smallDelta} \tau{,} {\smallDelta} \nu)
    \mathrm{e}^{{\shortminus}j 2 \pi(t \Delta \nu{\shortminus}f \Delta \tau{+}\tau \Delta f{\shortminus}\nu \Delta t)} \mathrm{d} t \mathrm{d} f \mathrm{d} \tau \mathrm{d} \nu \nonumber \\ 
    & {=} \sum\nolimits_{n{=1}}^\infty \lambda_n |\psi_n( \tau,\nu)|^2 |\phi_n(t, f)|^2 
\end{align}
where 
$|\psi_n( \tau,\nu)|^2$ and $|\phi_n(t, f)|^2$ represent the spectral density of $\psi_n(t,f)$ and $\phi_n(\tau,\nu)$, respectively.
%
Then, the 
\textit{global (or average) scattering function} $\overline{C}_H(\tau, \nu)$ and (local) TF path gain $\rho_H^2 (t,f)$ \cite{Matz2005NS} are calculated in \eqref{eq:GSF} and \eqref{eq:rho},
%
\begin{align}
\label{eq:GSF}
    & \overline{C}_H(\tau, \nu) {=} \mathbb{E}\{|S_H(\tau, \nu)|^2\} = \iint \mathcal{C}_H(t, f; \tau, \nu) ~dt ~df  \\
    & \rho_H^2 (t,f) {=}  \mathbb{E}\{|L_H(t, f)|^2\} = \iint \mathcal{C}_H(t, f; \tau, \nu) ~d\tau ~dv 
\label{eq:rho}
\end{align}
\eqref{eq:GSF} and \eqref{eq:rho} are re-expressed in terms of the spectral density of eigenfunctions by using \eqref{eq:LSF_eigen} and the properties in \eqref{eq:properties},  

\begin{align}
    &\overline{C}_H(\tau, \nu) {=} \mathbb{E}\{|S_H(\tau, \nu)|^2\} {=} \sum\nolimits_{n{=1}}^\infty \lambda_n |\psi_n( \tau,\nu)|^2 \\
    &  \rho_H^2 (t,f) {=} \mathbb{E}\{|L_H(t, f)|^2\} {=} \sum\nolimits_{n{=1}}^\infty \lambda_n |\phi_n(t, f)|^2 
\end{align}

Finally, the \textit{total transmission gain} $\mathcal{E}_H^2$ is obtained by integrating the LSF out with respect to all four variables, 
\begin{align}
\label{eq:TFgain}
    & \mathcal{E}_H^2 = \iiiint \mathcal{C}_H(t, f; \tau, \nu) ~dt ~df ~d \tau ~d\nu = \sum\nolimits_{n{=}1}^ \infty \lambda_n
\end{align}
\end{proof}
Consequently, the statistics of any wireless channel can be expressed by the eigenfunctions and eigenvalues obtained by the decomposition in \eqref{eq:channel_kernel_decomp}.
Therefore, we refer to Corollary \ref{col:characterization} as a \textit{unified characterization} of wireless channels. 
Further, Corollary~\ref{col:characterization} also suggests that the non-stationary channels are completely explained/characterized by the components decomposed by \eqref{eq:channel_kernel_decomp}, thereby serving as a validation of the correctness of HOGMT.
\section{Joint Spatio-Temporal Precoding} 
\label{sec:precoding}
\subsection{Dual space-time variation of non-stationary channels}

The received signal in~\eqref{eq:r_mu_discrete} can be expressed by the channel kernel as in \eqref{eq:MU}~\cite{MATZ20111}
\begin{align}
    \label{eq:MU}
    r_u (t) & = \int \sum\nolimits_{u'}  k_{u,u'}(t,t') s_{u'}(t') dt' + v_{u}(t)
\end{align}
where $v_u(t)$ is the noise, $s_u(t)$ is the data signal and $k_{u,u'} (t, t') {=} h_{u,u'}(t, t{-}t')$ is the channel kernel.
Then, the relationship between the transmitted and received signals is obtained by rewriting \eqref{eq:MU} in its continuous form in \eqref{eq:MU2}.
\begin{align}
    \label{eq:MU2}
    & r(u,t) = \iint k_H(u,t;u',t') s(u',t')~du'~dt' + v(u,t)
\end{align}
It is clear that this joint space-time interference (from both $u'$ and $t'$) varies along the space and time dimensions (i.e., across both $u$ and $t$). 
This is referred to as the \textit{dual space-time variation property} and it indicates that precoding using arbitrary joint space-time orthogonal basis is not sufficient to ensure interference-free communication, unless these basis remain orthogonal after propagating through the channel (as shown in Lemma \ref{lemma:bestprojection}).


Let $x(u,t)$ be the precoded signal, then the corresponding received signal is $Hx(u,t)$. The aim of precoding in this work is to minimize all existing interference of the channel, i.e., to minimize the least square error, $\|s(u,t) {-} Hx(u,t)\|^2$.
\begin{lemma}
\label{lemma:bestprojection}
Given a non-stationary channel $H$ with kernel $k_H(u{,}t{;}u'{,}t')$, if each projection in $\{H \varphi_n(u{,}t)\}$ are orthogonal to each other, there exists a precoded signal 
scheme 
$x(u,t)$ that ensures 
interference-free communication at the receiver, 
\begin{align}
    & \|s(u,t) - Hx(u,t)\|^2 = 0 
\label{eq:obj}
\end{align}
where $\varphi_n(u,t)$ is the 2-D eigenfunction of $x(u,t)$, obtained by KLT decomposition as in \eqref{eq:varphi_dcomp}
\begin{equation}
\label{eq:varphi_dcomp}
    x(u,t) = \sum_{n=1}^ \infty x_n \varphi_n(u,t)
\end{equation}
where $x_n$ is a random variable with $E\{x_n x_{n'}\}{=} \lambda_n \delta_{nn'} $. 
\vspace{-5pt}
\end{lemma}
\begin{proof}

$H\varphi_n(u,t) $ is the projection of $k_H(u,t;u',t')$ onto $\varphi_n(u',t')$ denoted by $ c_n(u,t)$ and is given by,
\begin{equation}
\label{eq:projections_c_n}
    c_n(u,t) =  \iint k_H(u,t;u',t') \varphi_n(u',t') ~du' ~dt'
\end{equation}

Using the above, \eqref{eq:obj} is expressed as,
\begin{align}
\label{eq:obj_trans}
     & ||s(u,t) - Hx(u,t)||^2 = ||s(u,t) - \sum_n^ \infty x_n c_n(u,t)||^2
\end{align}

Let $\epsilon (x) {=} ||s(u,t) - \sum_n^ \infty x_n \varphi_n(u,t)||^2$. Then its expansion is given by,
\begin{align}
\label{eq:ep}
     & \epsilon (x) = \langle s(u,t),s(u,t) \rangle - 2\sum_n ^ \infty  x_n \langle c_n(u,t),s(u,t) \rangle \\
     & + \sum_n^ \infty x_n^2 \langle c_n(u,t), c_n(u,t) \rangle \nonumber + \sum_n^ \infty \sum_{n' \neq n}^ \infty x_n x_{n'}  \langle c_n(u,t), c_{n'}(u,t) \rangle
\end{align}

Then the solution to achieve minimal $\epsilon(x)$ is obtained by solving for $\pdv{\epsilon(x)}{x_n} = 0$ as in \eqref{eq:solution}.
\begin{align}
\label{eq:solution}
    x_n^{opt} & {=}  \frac{\langle s(u,t), c_n(u,t) \rangle - \frac{1}{2} \sum_{n'\neq n}^ \infty x_{n'} \langle c_{n'}(u,t), c_n(u,t) \rangle }{\langle c_n(u,t), c_n(u,t) \rangle}
\end{align}
where $\langle a(u,t), b(u,t) \rangle {=} \iint a(u,t) b^*(u,t) ~du ~dt$ denotes the inner product. 
Let $\langle c_{n'}(u,t), c_n(u,t) \rangle = 0$, i.e., the projections $\{ c_n(u,t)\}_n$ are orthogonal basis. Then we have a closed form expression for $x^{opt}$ as in \eqref{eq:x_opt}.
\begin{align}
\label{eq:x_opt}
    x_n^{opt} & =  \frac{\langle s(u,t), c_n(u,t) \rangle}{\langle c_n(u,t), c_n(u,t) \rangle}
\end{align}

Substitute \eqref{eq:x_opt} in \eqref{eq:ep}, it is straightforward to show that $\epsilon(x){=} 0$.
\end{proof}



%

\subsection{Dual jointly orthogonal space-time decomposition}
Lemma~\ref{thm:thm2} formalizes the requirements for the joint space-time orthogonal basis $\{\varphi_n(u,t)\}$ to achieve a precoding scheme that ensures interference-free reception. From Theorem~\ref{thm:hogmt}, the 4-D channel kernel is decomposed as,

\begin{align}
\label{eq:thm2_decomp}
&k_H(u,t;u',t') = \sum\nolimits_{n{=1}}^\infty  \sigma_n \psi_n(u,t) \phi_n(u',t')
\end{align}
with properties as in \eqref{eq:thm2_decomp_pty},
\begin{align}
\label{eq:thm2_decomp_pty}
\begin{aligned}
&\iint \psi_n(u, t) \psi_{n'}^*(u, t) ~du ~dt {=} \delta_{nn'}  \\
&\iint \phi_n(u', t') \phi_{n'}^*(u', t') ~du' ~dt' {=} \delta_{nn'}
\end{aligned}
\end{align}

\eqref{eq:thm2_decomp} and \eqref{eq:thm2_decomp_pty} suggest that the 4-D kernel is decomposed into jointly orthogonal subchannels, $\{\psi_n(u,t)\}$ and $\{\phi_n(u',t')\}$. Moreover, combining \eqref{eq:thm2_decomp} and \eqref{eq:thm2_decomp_pty} leads to \eqref{eq:them2_duality}, which shows the \textit{duality} of the subchannels.
\begin{align}
    \iint k_H(u,t;u',t') \phi_n^*(u', t') ~du' ~dt' {=} \sigma_n \psi_n(u, t).
    \label{eq:them2_duality}
\end{align}
This duality suggests that when $\{\phi_n\}$ is transmitted through the 4-D channel, it transforms it to $\{\psi_n\}$ with random variables $\{\sigma_n\}$. Therefore, we refer to $\phi_n$ and $\psi_n$ as a pair of \textit{dual} eigenfunctions. Meanwhile, the transformation of $\phi_n$ to $\psi_n$ is scaled by $\sigma_n$ meaning that decomposed dual joint space-time orthogonal subchannels are flat-fading.


\subsection{HOGMT-based precoding}
Lemma 2 suggest precoding using $\{\varphi_n\}{=}\{\phi_n\}$ i.e., constructing $x(u,t)$ using $\{\phi_n\}$ with optimally derived coefficients $x_n$ using inverse KLT, eventually leads to interference-free communication, 
as the projections of channel kernels onto $\{\phi_n\}$ is $\{\sigma_n \psi_n\}$, which satisfies the orthogonal projection requirements in lemma~2 due to the orthogonal properties in \eqref{eq:thm2_decomp_pty}.

\noindent
\fbox{\begin{minipage}{0.97\linewidth}
\begin{theorem}
\label{thm:thm2}
(HOGMT-based precoding) Given a non-stationary channel $H$ with kernel $k_H(u,t;u',t')$, the precoded signal $x(u,t)$ that ensures interference-free communication at the receiver is constructed by inverse KLT as,
\begin{align}
\label{eq:x}
    x(u{,}t) {=} \sum_{n{=}1}^\infty x_n \phi_n^*(u{,}t),\text{where},
    x_n {=} \frac{{\langle} s(u{,}t){,} \psi_n(u{,}t) {\rangle}}{\sigma_n}
\end{align}
where 
$\{ \sigma_n \}$, $\{ \psi_n\}$ and $\{ \phi_n\} $ are obtained by decomposing the kernel $k_H(u,t;u',t')$ using Theorem~\ref{thm:hogmt} as in \eqref{eq:thm2_decomp}.
\end{theorem}
\end{minipage}}
\begin{proof}
The 4-D kernel $k_H(u,t;u',t')$ is decomposed into two separate sets of eigenfunction $\{\phi_n(u',t')\}$ and $\{\psi_n(u, t) \}$ using Theorem 1 as in \eqref{eq:thm2_decomp}. By transmitting the conjugate of the eigenfunctions, $\phi_n(u,t)$ through the channel $H$, we have that,  
\begin{align}
\label{eq:eigen_trans}
   H \phi_n^*(u,t) &{=} \iint k_H(u,t;u',t') \phi_n^*(u',t') ~du' ~d t' {=} \sigma_n \psi_n(u,t)
\end{align}
where $\psi_n(u,t)$ is also a 2-D eigenfunction with the orthogonal property as in \eqref{eq:thm2_decomp_pty}. From Lemma~\ref{lemma:bestprojection}, if the projection $c_n(u,t)$ in \eqref{eq:projections_c_n} is $c_n(u,t) = \sigma_n \psi(u,t)$, which satisfies the orthogonality $\langle \sigma_{n'}\psi_{n'}(u,t), \sigma_n\psi_n(u,t) \rangle = 0$, we achieve the optimal solution as in \eqref{eq:x_opt}. Therefore, let $x(u,t)$ be the linear combination of $\{\phi_n^*(u,t)\}$ with coefficients $\{x_n\}$ as in \eqref{eq:construct},

\begin{equation}
\label{eq:construct}
    x(u,t) = \sum_n^ \infty x_n \phi_n^*(u,t) 
\end{equation}

Then \eqref{eq:obj_trans} is rewritten as in \eqref{eq:obj_trans2},
\begin{align}
\label{eq:obj_trans2}
     & ||s(u,t) - Hx(u,t)||^2 = ||s(u,t) - \sum_n^ \infty x_n \sigma_n \psi(u,t)||^2
\end{align}
   
Therefore, optimal $x_n$ in \eqref{eq:x_opt} is obtained as in \eqref{eq:opt},

\begin{align}
    x_n^{opt} = \frac{\sigma_n \langle s(u,t), \psi_n(u,t) \rangle}{\lambda_n \langle \psi_n(u,t), \psi_n(u,t) \rangle} = \frac{\langle s(u,t), \psi_n(u,t) \rangle}{\sigma_n} \label{eq:opt}
\end{align}

Substituting \eqref{eq:opt} in \eqref{eq:construct}, the transmit signal is given by \eqref{eq:x_opt2},
\begin{equation}
\label{eq:x_opt2}
    x(u,t) = \sum_n^ \infty \frac{\langle s(u,t), \psi_n(u,t) \rangle}{\sigma_n} \phi_n^*(u,t). 
\end{equation}
\end{proof}

 As $\{\psi_n(u,t)\}$ are joint space-time orthogonal basis, data signal $s(u,t)$ can be expressed by 
\begin{align}
\label{eq:s}
    s(u{,}t) {=} \sum_{n{=}1}^\infty s_n \psi_n(u{,}t),\text{where},
    s_n {=} {\langle} s(u{,}t){,} \psi_n(u{,}t){\rangle}
\end{align}

Therefore, the precoding in Theorem \ref{thm:thm2} can be explained as transmitting the eigenfunctions $\{\phi_n^*(u,t)\}$ after multiplying with derived coefficients $\{x_n\}$, which will transfer to $\{s_n\}$\footnote{Although precoding involves a linear combination of $\phi_n^*(u,t)$ with $x_n$. Generally, HOGMT-precoding is a non-linear function ($\mathcal{W}(\cdot)$) with respect to the data signal $s(u,t)$, i.e., $x(u,t) {=} \mathcal{W}(k_H; s(u,t))$.}.
Then the data signal $s(u,t)$ is directly reconstructed at the receiver by the dual eigenfunctions $\{\psi_n(u,t)\}$ with transferred coefficients $\{s_n\}$ to the extent of noise $v_u(t)$ as the net effect of precoding and propagation in the channel ensures that from \eqref{eq:MU}, $r(u,t){=}Hx(u,t){+}v_u(t){\rightarrow}s(u,t){+}v_u(t){=}\hat{s}(u,t)$ using Lemma \ref{lemma:bestprojection}, 
where $\hat{s}(u,t)$ is the estimated signal. 
Therefore, the spatio-temporal decomposition of the channel in Theorem \ref{thm:hogmt} allows us to precode the signal such that all interference in the spacial domain, time domain and joint space-time domain are cancelled when transmitted through the channel, leading to a joint spatio-temporal precoding scheme. 
Further, this precoding ensures that the modulated symbol is reconstructed directly at the receiver with an estimation error that 
of $v_u(t)$ (equation (13) in section III), 
thereby completely pre-compensating the spatio-temporal fading/ interference in non-stationary channels to the level of AWGN noise.
Therefore, this precoding does not require complementary step at receiver, which vastly reducing its hardware and computational complexity compared to state-of-the-art precoding methods like Dirty Paper Coding (DPC) or linear precoding (that require a complementary decoder [45].

\subsection{Geometric interpretation of precoding}
\begin{figure}[h]
    \centering
    \includegraphics[width=0.5\linewidth]{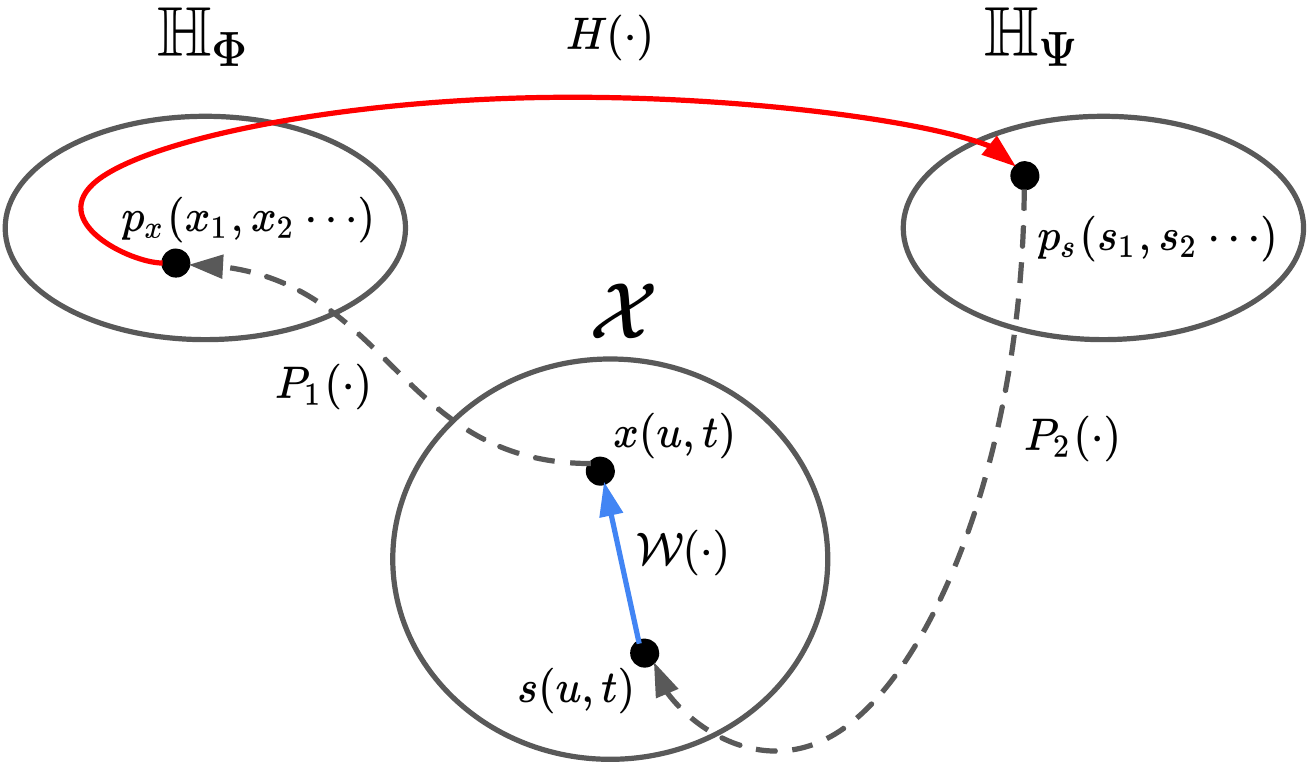}
    \caption{Geometric interpretation of HOGMT-precoding}
    \label{fig:Geometric_interpretation}
\end{figure}
Figure~\ref{fig:Geometric_interpretation} shows a geometric interpretation of HOGMT precoding without considering the noise item. Given two Hilbert Space $\mathbb{H}_\Phi$ and $\mathbb{H}_\Psi$, where basis are eigenfunctions $\{\phi_n\}$ and $\{\psi_n\}$, respectively, the precoded signal $x(u,t) {=} \mathcal{W}(k_H; s(u,t))\in \mathcal{X}$ can be seen as a point $p_x(x_1, x_2, \cdots) = P_1(x(u,t)) \in \mathbb{H}_\Phi$, where $x_n$ obtained in \eqref{eq:x} is $n$\textsuperscript{th} coordinate. Then the 4-D channel have the transform $H(\cdot): p_x(x_1, x_2, \cdots) \in \mathbb{H}_\Phi \to p_s(s_1, s_2, \cdots) \in \mathbb{H}_\Psi$, where the point $p_s(s_1, s_2, \cdots)$ represented in reality (project to space-time space $\mathcal{X}$) is directly the data signal $s(u,t) = P_2(p_s(s_1, s_2, \cdots)) \in \mathcal{X}$. The dual spatio-temporal variation of the 4-D channel not only transfer the coordinate $(x_1, x_2, \cdots)$ to $(s_1, s_2, \cdots)$, but also transfer Hilbert Space $\mathbb{H}_\Phi$ to $\mathbb{H}_\Psi$. As HOGMT extract the duality (explain $H(\cdot)$) and dual orthogonality (explain $P_1(\cdot)$ and $P_2(\cdot)$) in this dual variation, we can use inverse method to construct the precoded signal, i.e., for the target closed loop $\mathcal{W}(\cdot) + P_1(\cdot) + H(\cdot) + P_2(\cdot) = 0$, we have $\mathcal{W}(\cdot) = - P_2(\cdot) - H(\cdot) - P_1(\cdot)$, meaning, $\mathcal{W}(\cdot)$ can be obtained by the inverse process $s(u,t) \to p_s(s_1, s_2, \cdots) \to p_x(x_1, x_2, \cdots) \to x(u,t)$, which is equivalent to \eqref{eq:x}.

\section{Implementation of HOGMT Precoding}
\label{sec:implementation}

HOGMT decomposition is the most important and computational part for HOGMT precoding. Ideally, nonlinear approximation by eigenfunctions is optimal in terms of mean square errors~\cite{cohen1997nonlinear}. However, extraction of eigenfunctions is very undesirable~\cite{Liu2004Eigen}. There is no direct method to implement nonlinear HOGMT because of limitation of tools. Here we consider the linear alternatives.   
\subsection{Equivalent tensor form}

\label{sec:system}
Denote the tensor form of $H(t,\tau)$ in \eqref{eq:H_t_tau} as $\mathbf{H} \in \mathbb{C}^{L_u \times L_t \times L_{u'} \times L_{\tau}}$, where $L_u$ and $L_{u'}$ are the number of users and transmit antennas, and assume $L_u \geq L_{u'}$.
$L_t$ and $L_{\tau}$ are the number of data symbols and delay taps. Then the 4-D kernel tensor $\mathbf{K} \in \mathbb{C}^{L_u \times L_t \times L_{u'} \times L_{t'}}$ is obtained by shifting coordination of $\mathbf{H}$, where $L_{t'} = L_{t}$. HOGMT in \eqref{eq:thm2_decomp} decompose the 4-D process into 2-D eigenfunctions, meaning it decompose the 4-D tensor into 2-D jointly orthogonal blocks (eigenmatrices) as in \eqref{eq:tensor_channel_kernel_decomp}


\begin{align}
\label{eq:tensor_channel_kernel_decomp}
&\mathbf{K} = \sum_n \sigma_n \Psi_n \otimes \Phi_n
\end{align}
where $\otimes$ is Kronecker product. $\Psi_n \in \mathbb{C}^{L_u \times L_t}$ and $\Phi_n \in \mathbb{C}^{L_{u'} \times L_{t'}}$ are eigenmatrices with properties in \eqref{eq:tensor_property},
\begin{align}
\label{eq:tensor_property}
\begin{aligned}
\langle \Psi_n, \Psi_{n'}^* \rangle_F = \delta_{nn'} \quad \text{and,} \quad
\langle \Phi_n, \Phi_{n'}^* \rangle_F = \delta_{nn'}
\end{aligned}
\end{align}
which is Frobenius product form of \eqref{eq:properties}. Then the duality in \eqref{eq:them2_duality} is transfered to \eqref{eq:tensor_dual},

\begin{align}
\label{eq:tensor_dual}
\langle \mathbf{K}, \Phi_{n}^* \rangle_F = \sigma_n \Psi_n
\end{align}

The transmit space-time signal block is thus the combination of eigenmatrices with optimally derived coefficients from Theorem~\ref{thm:thm2}. Higher-order SVD (HOSVD) is one choice to decompose tensor into eigenvectors at each dimension. Then $\Psi_n$ is Kroneker product of eigenvectors at $u$ and $t$ domain, whereas $\Phi_n$ is Kroneker product of eigenvecors at $u'$ and $t'$ domain. However, the HOSVD is extremely complex for 4-D tensor, especially for non-truncated kernels \cite{badeau2008fast}. 
Considering the proposed method just require the decorrelation of space-time domain at transmit and receiver instead of each dimension, we further proposed an implementable (and low-dimensional and low-complexity) alternative as in Lemma~\ref{lemma:Dimensionality_reduction_for_HOGMT}.
\label{sec:complexity}
\begin{figure}[t]
    \centering
    \includegraphics[width=0.6\linewidth]{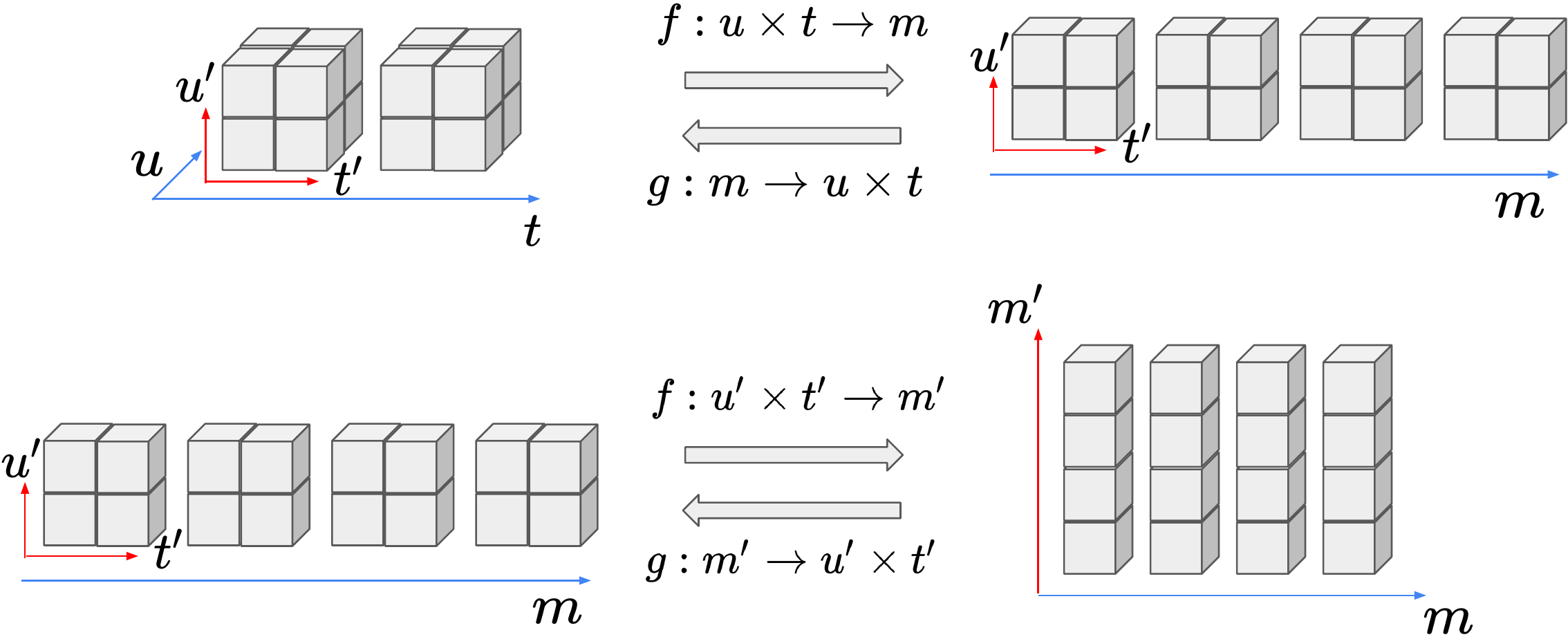}
    \caption{Illustration of dimensionality reduction}
    \label{fig:dimensionality_reduction}
\end{figure}
\begin{lemma}
\label{lemma:Dimensionality_reduction_for_HOGMT}
(Dimensionality reduction for HOGMT) Given a 4-D tensor $\mathbf{K} \in \mathbb{C}^{L_u \times L_t \times L_{u'} \times L_{t'}}$ and an invertible mapping $f{:}u{\times}t{\to}m$, let $\mathbf{K'} = f(\mathbf{K}) \in \mathbb{C}^{L_m \times L_{m'}}$. We have 
\begin{align}
    \mathbf{K'} = \mathbf{U\Sigma V^*} 
    = \sum_n z_n \mathbf{u}_n \otimes \mathbf{v^*}_n
\label{eq:2d_hosvd}
\end{align}
where $z_n$ is the singular value. $\mathbf{u}_n \in \mathbb{C}^{L_m \times 1}$ and $\mathbf{v}_n \in \mathbb{C}^{L_{m'} \times 1}$ are eigenvectors. 
For \eqref{eq:tensor_channel_kernel_decomp}, there exists the equivalent 
\begin{align}
    &\sigma_n {=} z_n, \quad \Psi_n {=} g(\mathbf{u}_n) \quad \text{and} \quad \Phi_n {=} g(\mathbf{v}_n^*)
\label{eq:equivalent}
\end{align}
where $g{:}m \to u{\times}t$ is the inverse mapping of $f$.
\end{lemma}

\begin{proof}
\begin{align}
& \mathbf{K} {=} g(\mathbf{K'}) = g(\mathbf{U\Sigma V^*}) = g(\sum_n z_n \mathbf{u}_n \otimes \mathbf{v^*}_n) \nonumber \\ 
& {=} \sum_n z_n g(\mathbf{u}_n) \otimes g(\mathbf{v^*}_n) \label{eq:maping_proof}
\end{align}
Substituting \eqref{eq:maping_proof} in \eqref{eq:tensor_channel_kernel_decomp}, we directly have the equivalent \eqref{eq:equivalent}.
\end{proof}

Figure~\ref{fig:dimensionality_reduction} shows the transition of a $2\times 2 \times 2 \times 2$ tensor $\mathbf{K}$ to a $4 \times 4$ matrix $\mathbf{K'}$. The linear mapping $f$ and $g$ are straightforward. 

\subsection{Computational complexity}

DPC incurs a much higher runtime complexity (Factorial complexity \cite{Mao2020DPC}) compared to HOGMT-precoding (polynomial complexity). Assuming $L_u \geq L_{u'}$, the complexity of HOGMT and DPC are given by Table~\ref{tab:Complexity},
\begin{table}[h]
\caption{Complexity Comparison}
\renewcommand*{\arraystretch}{1.2}
\centering
\begin{tabular}{|c|c|}
\midrule
\textbf{Strategy} & \textbf{Computational Complexity} \\\midrule
HOGMT with Lemma \ref{lemma:Dimensionality_reduction_for_HOGMT} & $\mathcal{O} (L_{u} L_{u'}^2 L_{t}^3)$      \\\midrule
HOGMT with HOSVD &  $\mathcal{O} ((\frac{L_{u}+L_{u'}+2L_{t}}{4})^5 + L_{u} L_{u'}L_{t}^2)$ \\\midrule
DPC& $\mathcal{O} (L_{t}((L_{u} L_{u'})^{3.5} + L_{u} L_{u'}^2) L_{u'}!)$  \\ \midrule
\end{tabular}
\label{tab:Complexity}
\end{table}

\begin{figure}[t]
\centering
\begin{subfigure}{.3\textwidth}
  \centering
  \includegraphics[width=1\linewidth]{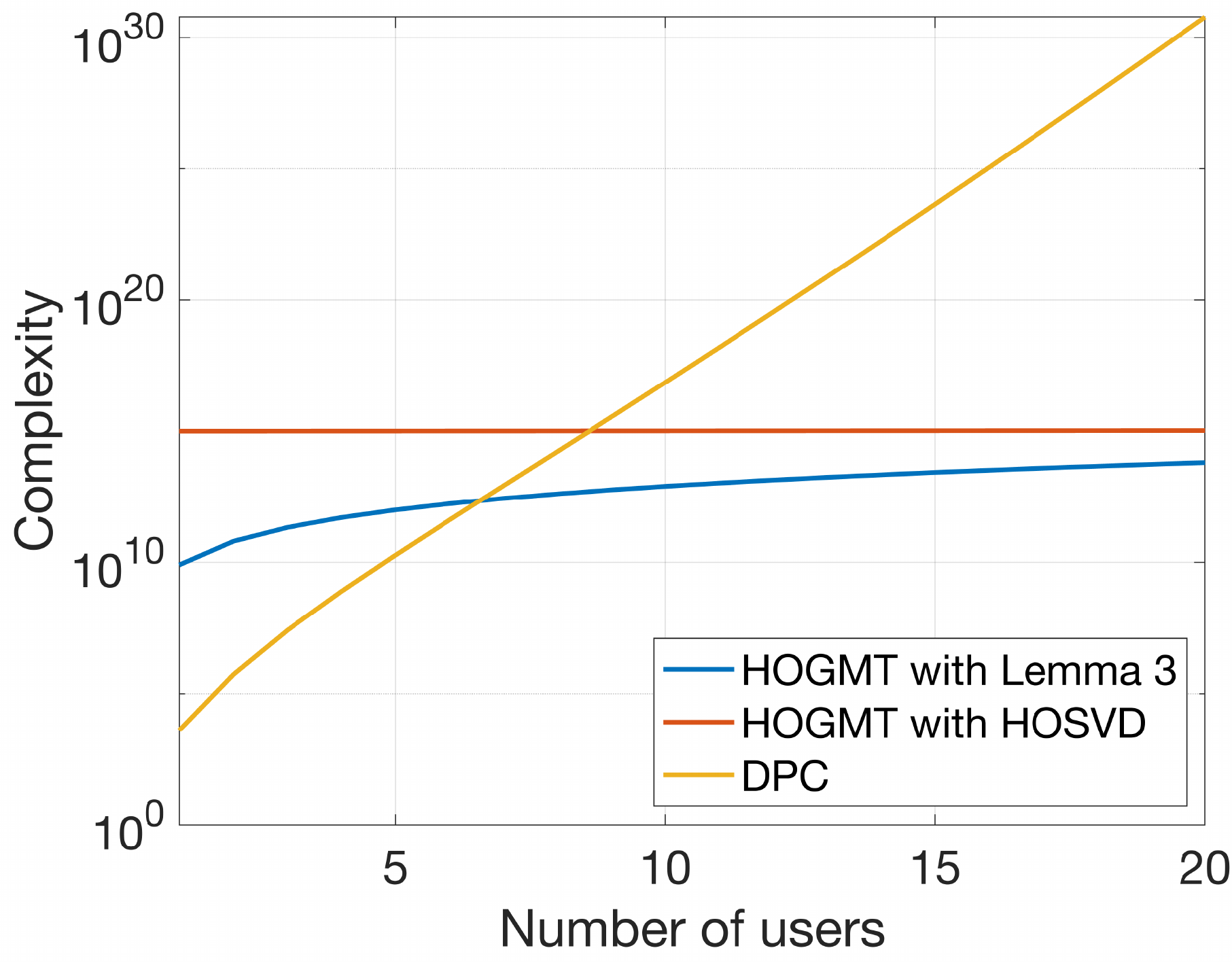}
  \caption{Complexity for $L_t = 2000$}
  \label{fig:complexity_1}
\end{subfigure}
\qquad
\begin{subfigure}{.3\textwidth}
  \centering
  \includegraphics[width=1\linewidth]{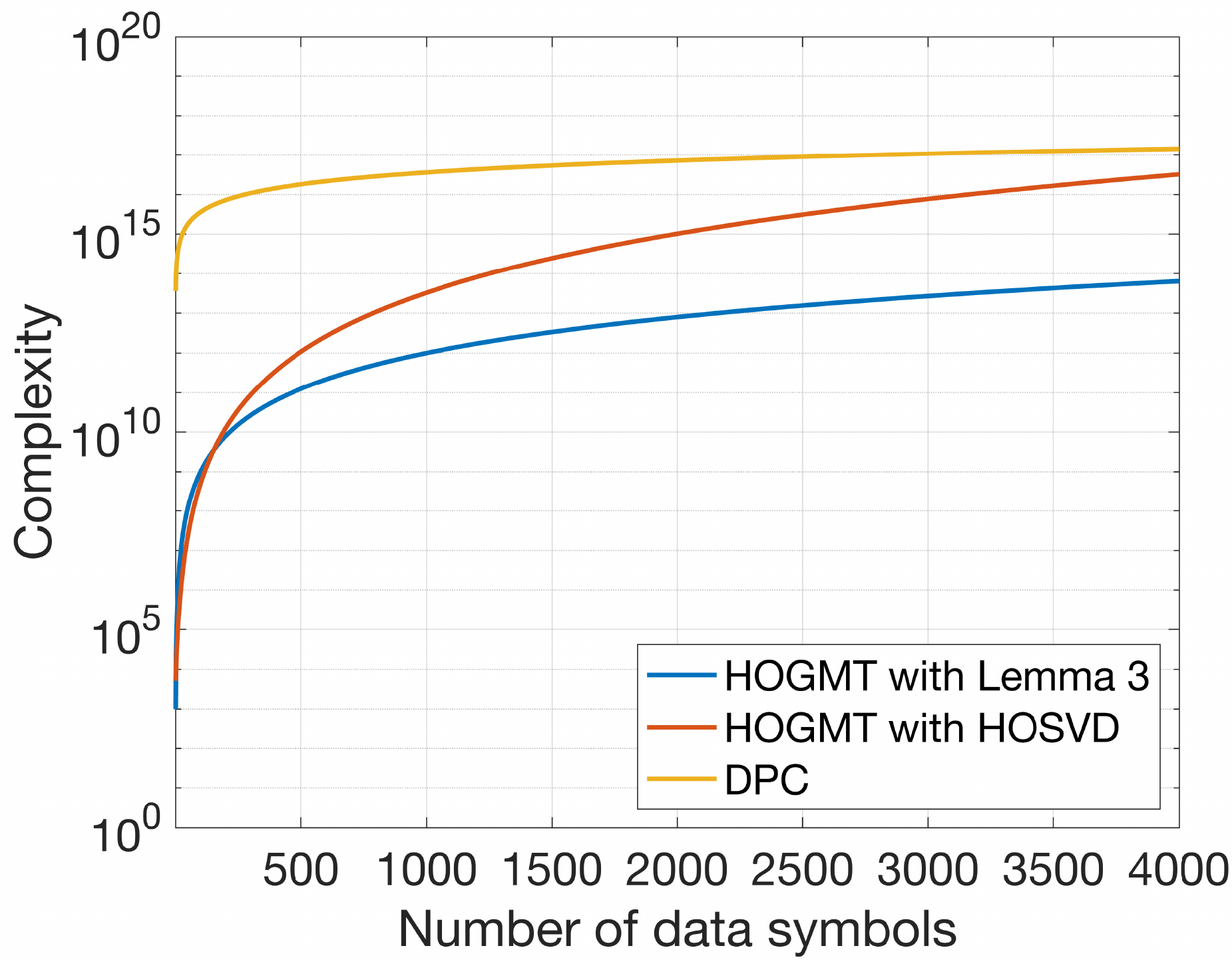}
  \caption{Complexity for $L_u = 10$}
  \label{fig:complexity_2}
\end{subfigure}
  \caption{Complexity comparison for HOGMT with two implementation methods and DPC }
  \label{fig:complexity}
\end{figure}

Figure~\ref{fig:complexity_1} compare the complexity of DPC and two implementations of HOGMT with $2000$ data symbols with respect to the number of users. When users are more than $8$, DPC is much more complex than HOGMT. HOGMT with Lemma~\ref{lemma:Dimensionality_reduction_for_HOGMT} has less complexity than HOGMT with HOSVD, though gap narrows as $L_u$ approaches $L_t$. However, the gap widens in Figure~\ref{fig:complexity_2}, as $L_t$ is further larger than $L_u$. 


\section{Evaluation Results}

\subsection{Practical non-stationary channel simulation framework}
We analyze the accuracy of the proposed joint spatio-temporal precoding using 3GPP 38.901 UMa NLOS senario built on QuaDriga in Matlab. The channel parameters and the layout of the base station (BS) and the user equipment (UE) are shown in Table~\ref{tab:parameters}.
The QuaDriga channel toolbox has been shown to accurately reflect realistic modern channels (\eg V2X, HST) using practical measurements in \cite{QuaDriGa}. Moreover, this simulated testbed gives the freedom to address a variety of adverse and different scenarios of non-stationary channels that may not be observed without extensive measurement campaigns.
\begin{table}[h]
\caption{Non-stationary channel parameters}
\renewcommand*{\arraystretch}{1.2}
\centering
\begin{tabular}{l}
\midrule
Senario: 3GPP 38.901 UMa NLOS~\cite{3gpp.38.901}; Bandwidth: $20$ Mhz; Center frequency: $5$ Ghz; Subcarriers : $64$ \\\midrule
BS layout: Array type:3GPP 3-D~\cite{3gpp.37.885}; Antenna height: 10 m; Antenna number: 10 \\\midrule
UE layout: Array type:vehicular~\cite{3gpp.36.873}; Antenna height: 1.5 m; UE number: 10; Speed: $120 \pm 18$ km/h \\\midrule
\end{tabular}
\label{tab:parameters}
\end{table}


\begin{figure*}[t]
\centering
\begin{subfigure}{.27\textwidth}
  \centering
  \includegraphics[width=1\linewidth]{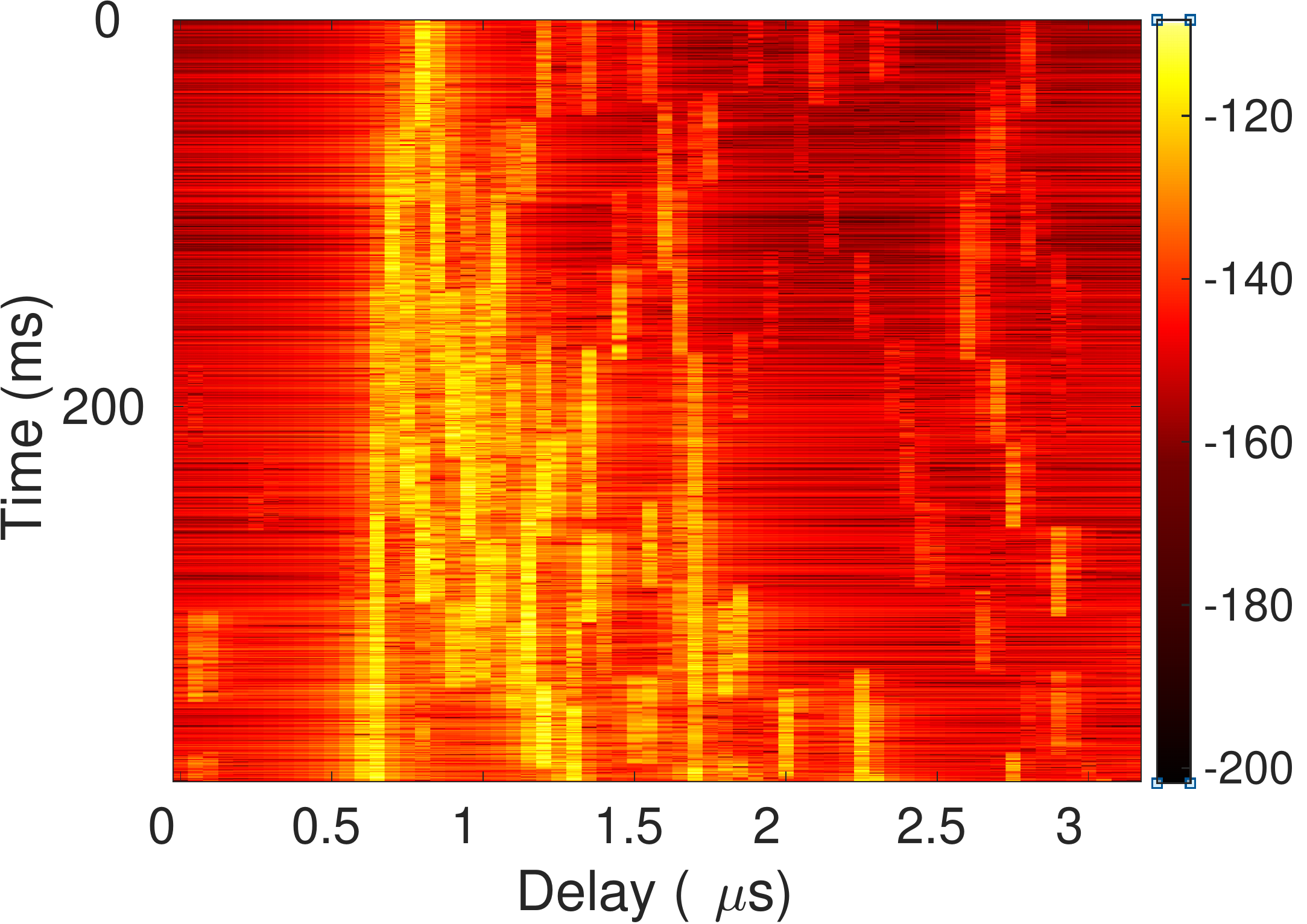}
  \caption{Power delay profile in dB}
  \label{fig:pdp}
\end{subfigure}
\qquad
\begin{subfigure}{.265\textwidth}
  \centering
  \includegraphics[width=1\linewidth]{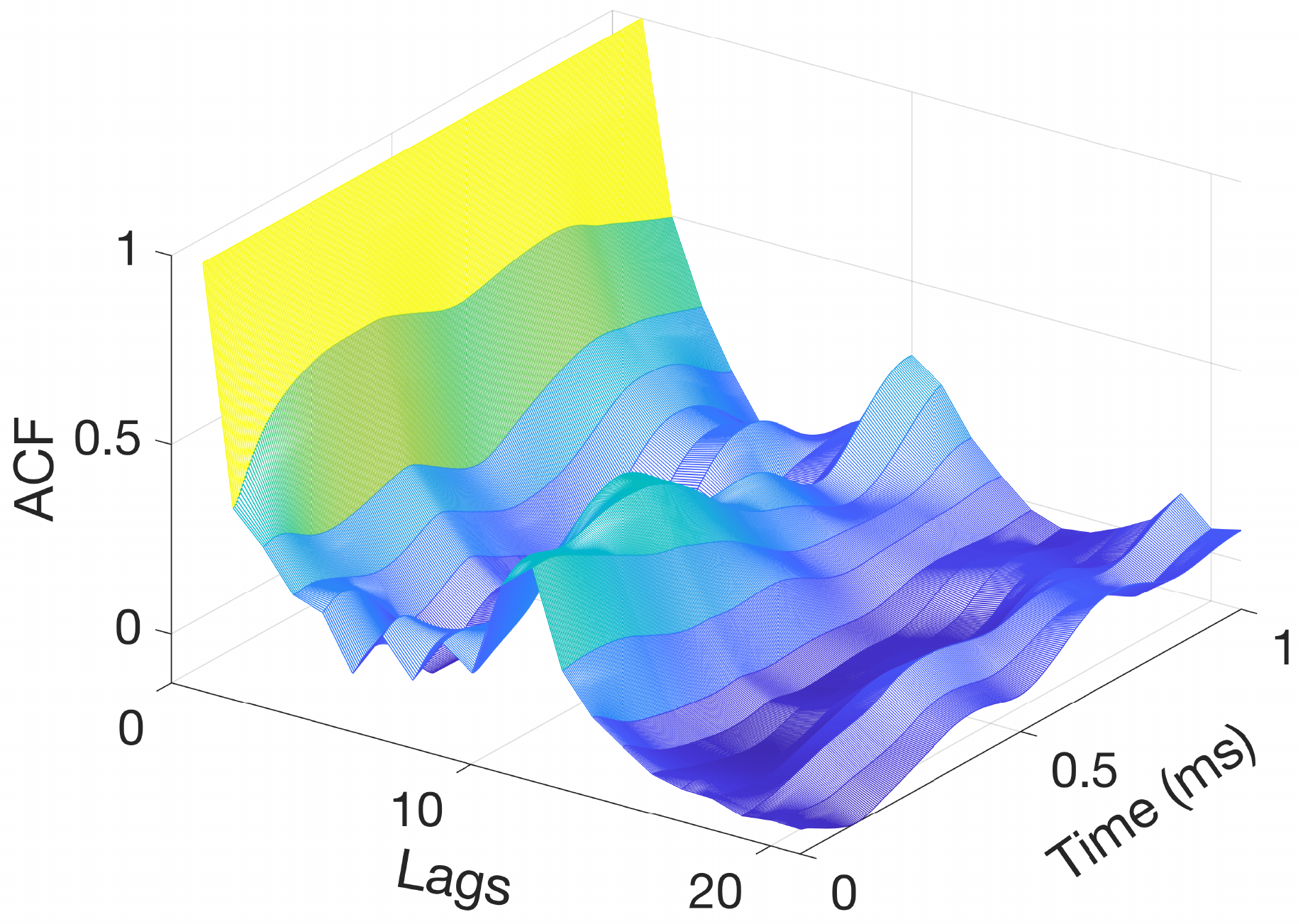}
  \caption{ACF for the first $1$ ms}
  \label{fig:Acf}
\end{subfigure}
\qquad
\begin{subfigure}{.26\textwidth}
  \centering
  \includegraphics[width=1\linewidth]{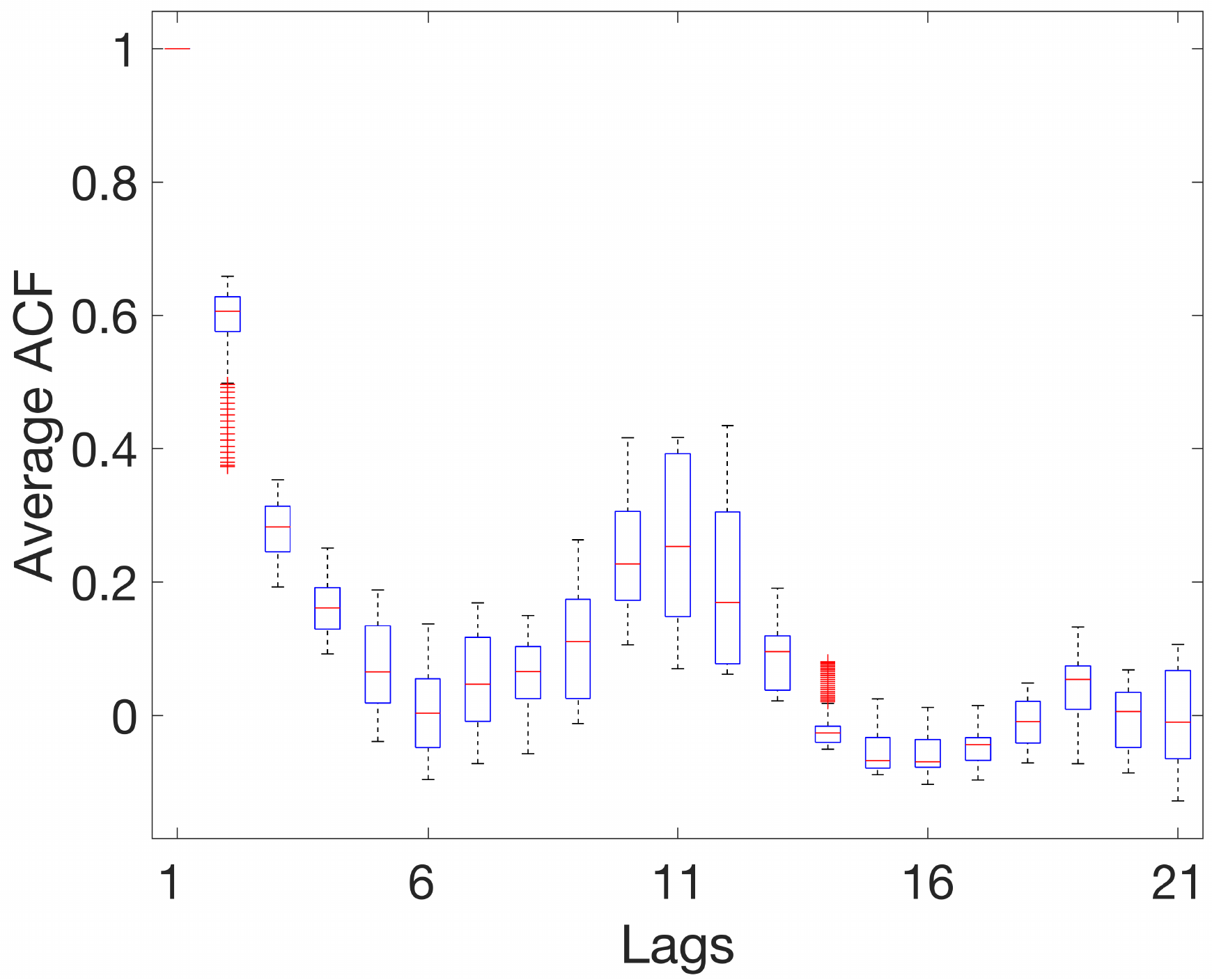}
  \caption{Temporal variation of ACF for the first $1$ ms}
  \label{fig:boxplot}
\end{subfigure}
  \caption{Power delay profile and statistics distribution for $h_{u,u'}(t,\tau)$ with $u = 1$ and $u' = 1$}
\label{fig:temporal_interference}
\end{figure*}

Figure~\ref{fig:pdp} shows the power delay profile from the antenna $u'=1$ to the user $u=1$ in terms of the 4-D channel i.e., $h_{u,u'}(t,\tau)$ , where there exist a drift due to the mobility of the user, leading to the time-varying distribution in overview. We further measure autocorrelation function (ACF) of time-varying impulse response $h_{1,1}(t,\tau)$ in Figure~\ref{fig:Acf} for the first $1$ ms. ACF changing over time corroborates its non-stationarity, even within $1$ ms. Figure~\ref{fig:boxplot} shows the distribution of statistics of ACF (mean and variance), meaning the non-stationarity degree to some extent.

Figure~\ref{fig:hs1} and figure~\ref{fig:hs2} show the spatial channel gains for two fixed time instance (time instances $t{=}1$ ms and $t{=}2$ ms both with $\tau = 0$), respectively. The channel gains from other users leads to the spatial (inter-user) interference. We observe that the spatial interference also changes over time. Figure~\ref{fig:temporal_interference} and figure~\ref{fig:spatial_profile} separately show the cause of time-varying temporal interference and spatial interference for $h_{u,u'}(t,\tau)$. 

\noindent
\textbf{Stationary interval:} The \textit{Correlation Matrix Distance (CMD)} is a measure for the degree of stationarity of narrowband MU/MIMO channels \cite{Renaudin2010NonStationaryNM} and is defined in \eqref{eq:CMD}.
\begin{equation}
    \label{eq:CMD}
    d_{corr}(t,\Delta t) = 1 - \frac{\langle \mathbf{R}(t),\mathbf{R}(t+\Delta t)\rangle_F} {||\mathbf{R}(t)||_F ||\mathbf{R}(t+\Delta t) ||_F}
\end{equation}
where $||\cdot||_F$ is the Frobenius norm and $\mathbf{R}$ is the correlation matrix. 
In our work, the CMD is calculated at both the transmitter (Tx) and Receiver (Rx) sides, and the corresponding correlation matrices 
for the narrowband channel $\Tilde{\mathbf{H}}(t) = \int \mathbf{H}(t,\tau) d \tau$ over period $T$ are given by \eqref{eq:corr_tx_rx}.
\begin{align}
\label{eq:corr_tx_rx}
    &\mathbf{R_{Tx}}(t) {=} \frac{1}{T}\int_{t}^{t+T} \Tilde{\mathbf{H}}(t)^T \Tilde{\mathbf{H}}(t)^* d t  \\ 
    &\mathbf{R_{Rx}}(t) {=} \frac{1}{T}\int_{t}^{t+T} \Tilde{\mathbf{H}}(t) \Tilde{\mathbf{H}}(t)^H d t \nonumber
\end{align}
Consequently, the time-varying stationary interval is defined as the largest duration over which CMD remains below a predefined threshold $d_0$, i.e., $\mathcal{T}(t) = |\Delta t_{max}(t) - \Delta t_{min}(t)|$ where,
\begin{align}
    &\Delta t_{max}(t) {=} \argmax_{\Delta t \leq 0} d_{corr}(t,\Delta t) \geq d_0 \nonumber\\
     &\Delta t_{min}(t) {=} \argmin_{\Delta t \geq 0} d_{corr}(t,\Delta t) \geq d_0 \nonumber
\end{align}

Figure \ref{fig:CMD_TX} and figure \ref{fig:CMD_RX} show the CMD at the Tx and Rx, respectively. We observe that, the stationary interval for $d_0 = 0.2$ and $d_0 = 0.3$ (in table \ref{tab:ns}, this threshold for V2X is 0.2, and for HST is 0.7 - 0.9.) are about $400$ $\mu$s and $500$ $\mu$s, respectively.
Here, the stationarity interval is presented in time instead of distance, as the varying mobility profiles of the multiple users lead to different distances over which stationarity holds. We observe a lower SI compared to that reported in table \ref{tab:ns}, due to the following reasons: a) we consider the multi-user channel, where the varying mobility profiles of spatial elements causes rapid changes in the spatial correlations, b) the rich and dynamic propagation environment in the 3GPP UMa NLOS channel model results in adverse scattering in the simulation environment compared to those observed in the reported experiments in table \ref{tab:ns}. However, as the precoding is able to achieve low BER for smaller SI (higher degree of NS), it would also ensure low BER for larger SI observed in the reported measurement campaigns.

\subsection{Performance of the proposed precoding}
\begin{figure*}[t]
\centering
\begin{subfigure}{.24\textwidth}
  \centering
  \includegraphics[width=1\linewidth]{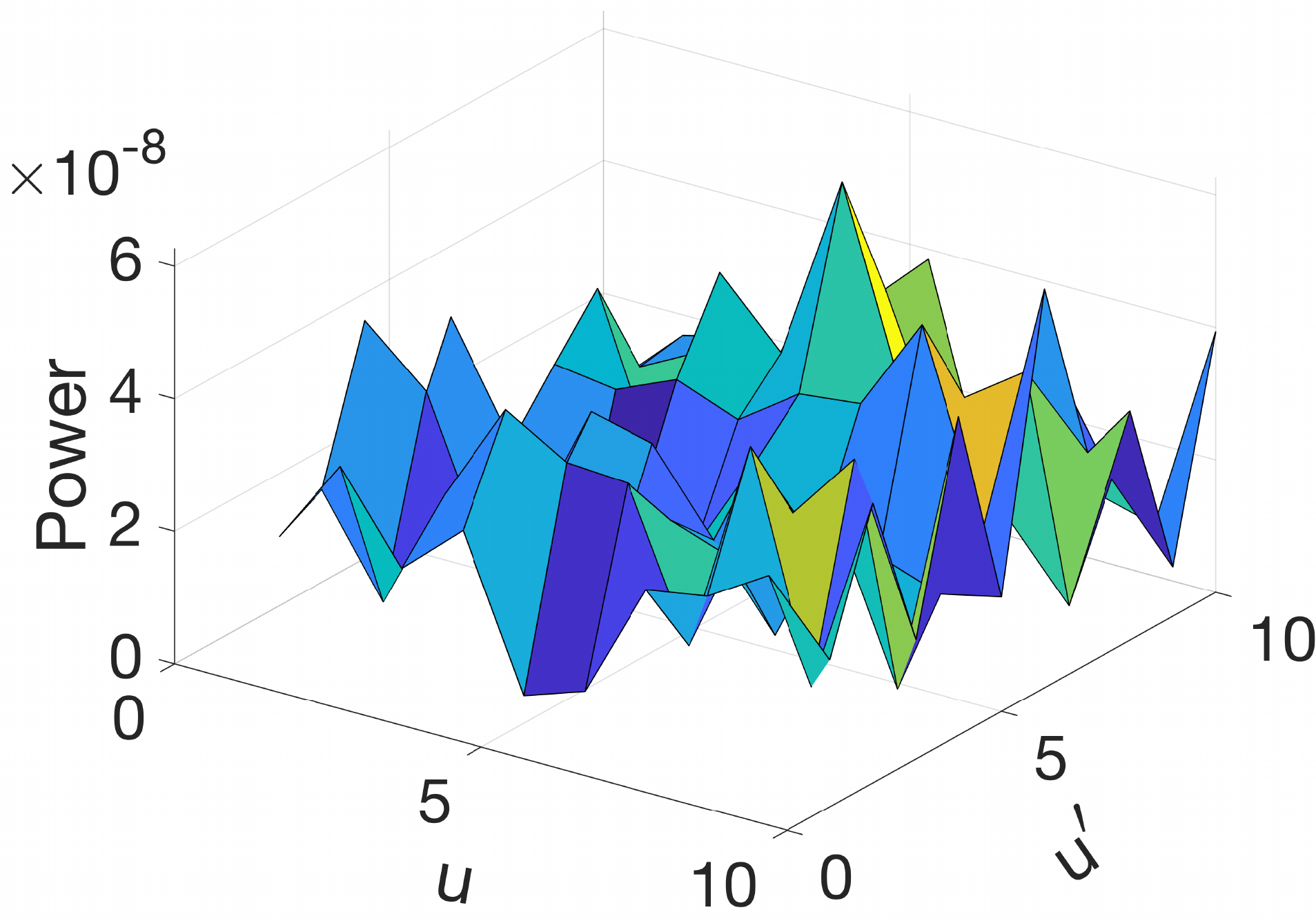}
  \caption{Spatial gain at $t = 1$ ms
  } 
  \label{fig:hs1}
\end{subfigure}
\begin{subfigure}{.24\textwidth}
  \centering
  \includegraphics[width=1\linewidth]{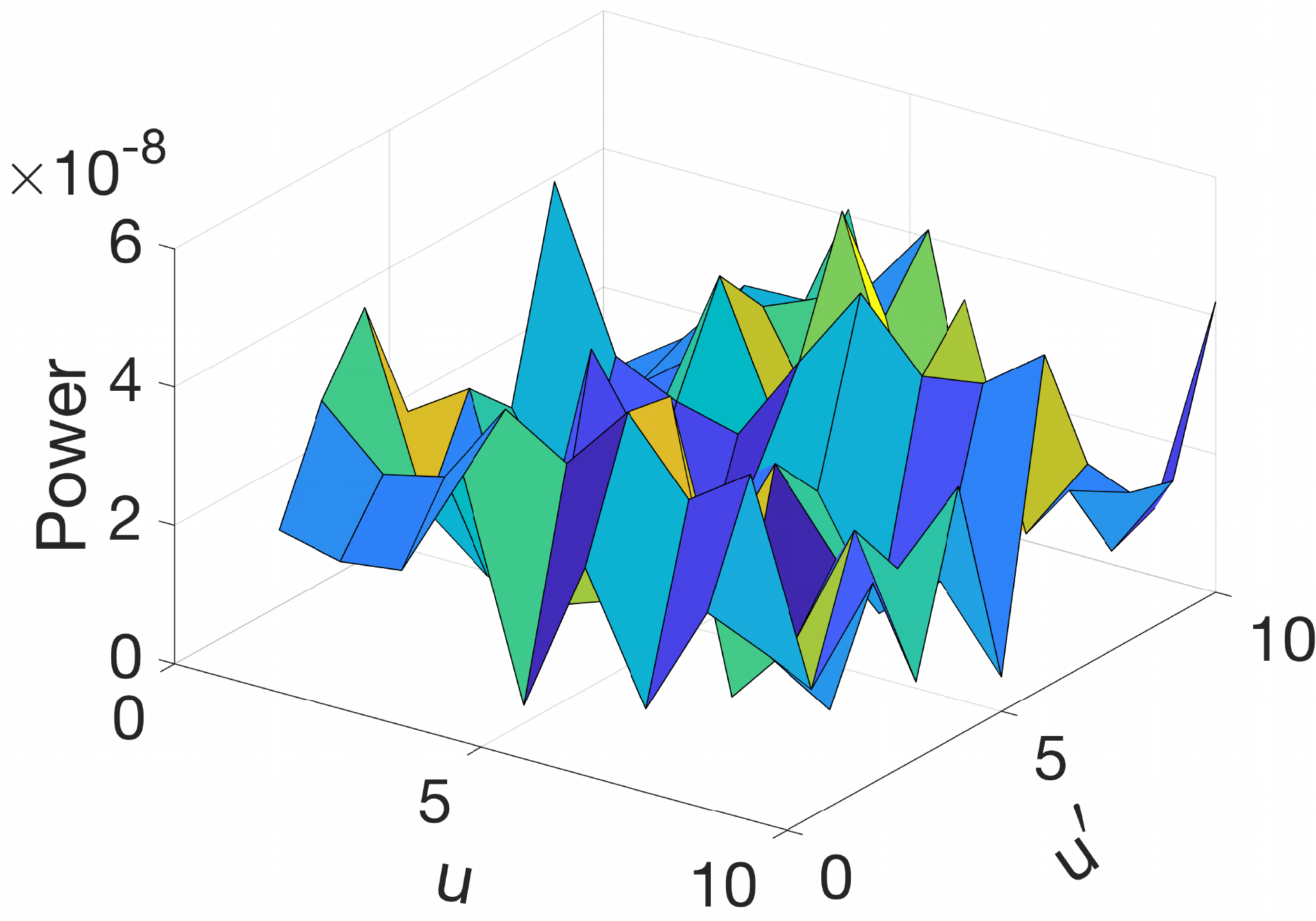}
  \caption{Spatial gain at $t = 2$ ms}
  \label{fig:hs2}
\end{subfigure}
\begin{subfigure}{.24\textwidth}
  \centering
\includegraphics[width=1\linewidth]{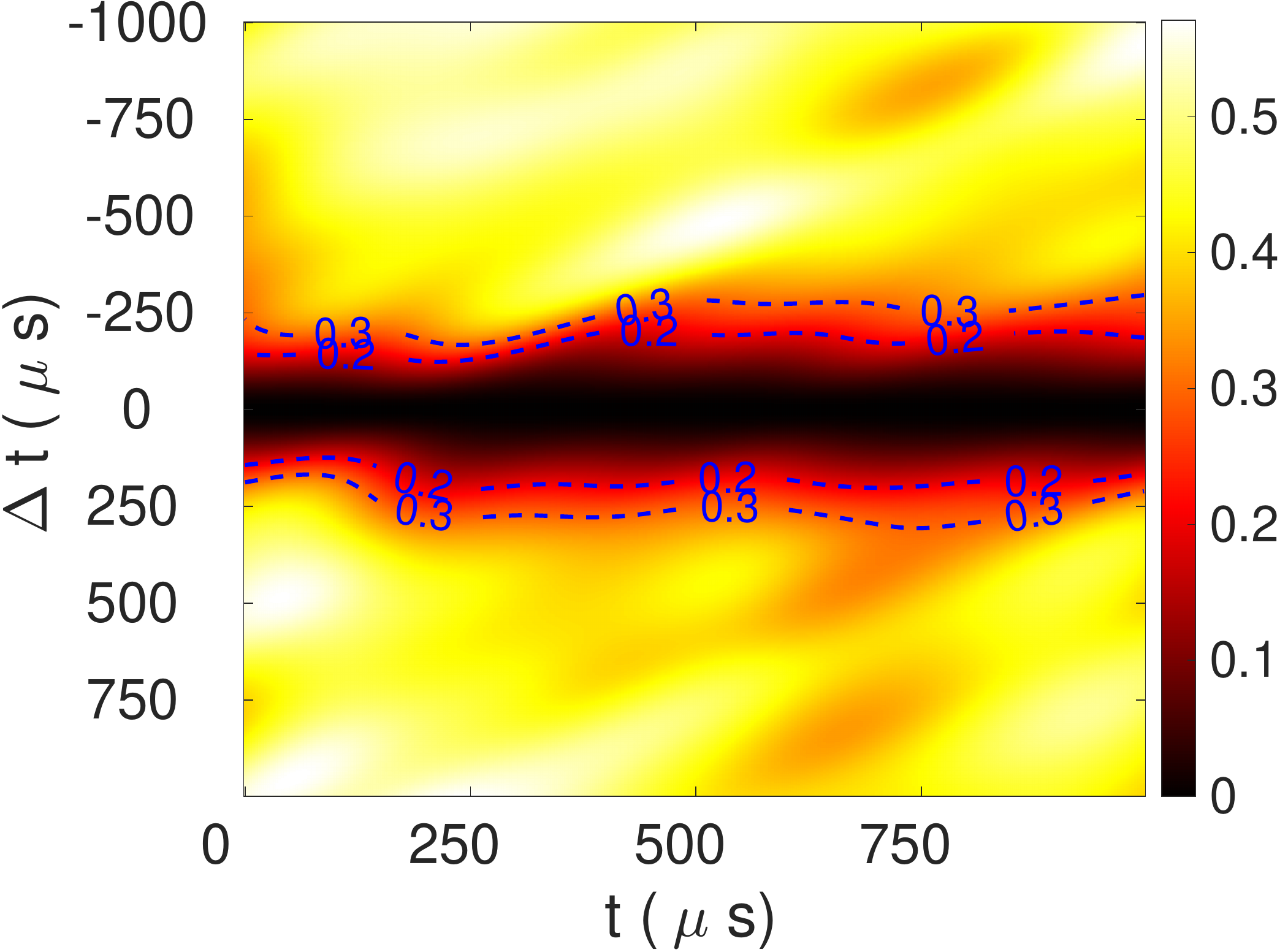}
  \caption{CMD at Tx }
  \label{fig:CMD_TX}
\end{subfigure}
\begin{subfigure}{.24\textwidth}
  \centering
  \includegraphics[width=1\linewidth]{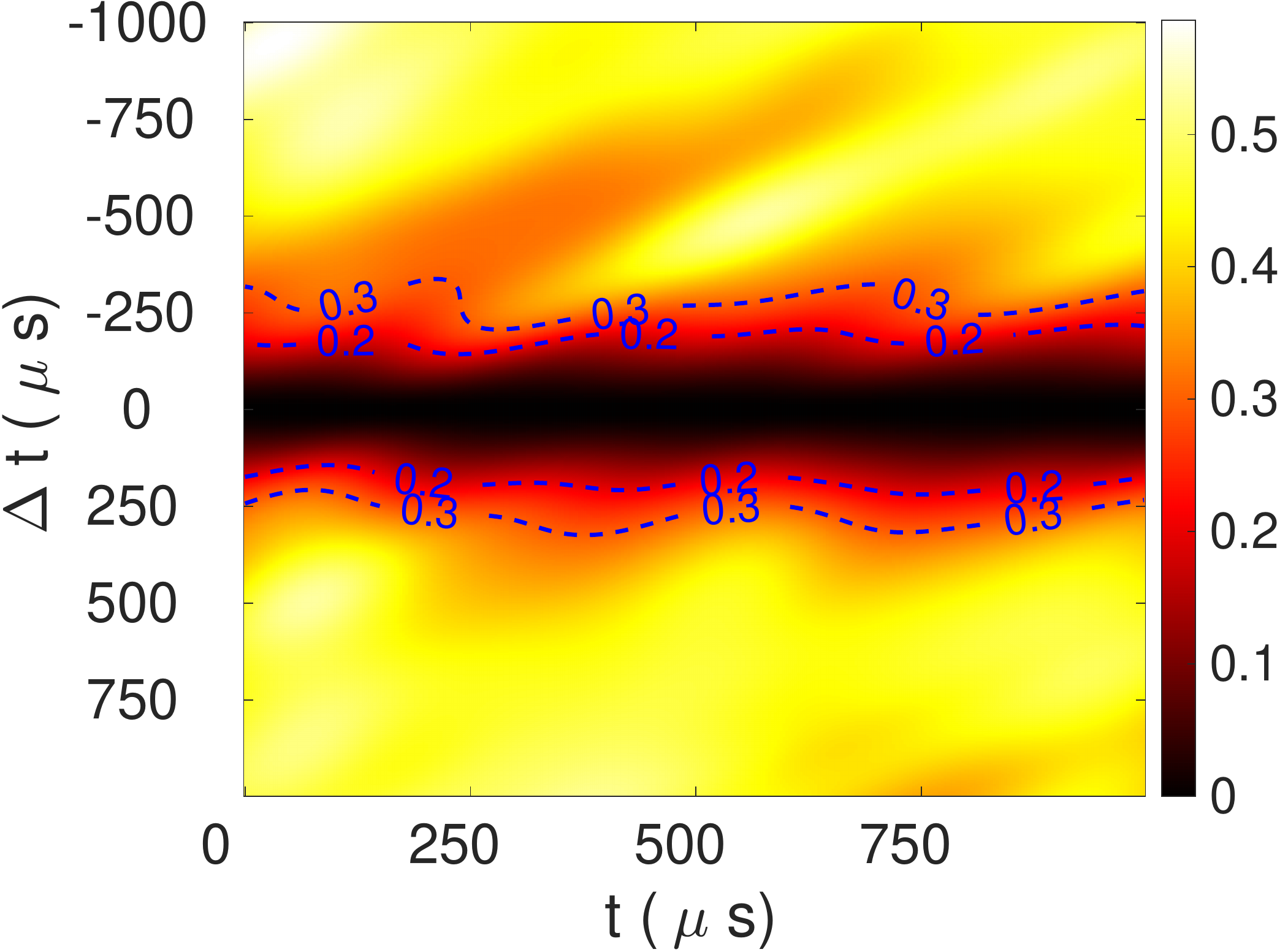}
  \caption{CMD at Rx}
  \label{fig:CMD_RX}
\end{subfigure}
\caption{Time-varying spatial channel gains and stationary intervals}
\label{fig:spatial_profile}
\end{figure*}


The proposed joint spatio-temporal precoding involves extracting the 4-D channel kernel $k_H(u,t;u',t') = h_{u,u'}(t,t-t')$. Figure~\ref{fig:4-D kernel} shows the 4-D channel kernel for $u = 1$ and $u = 2$ at $t = 1000$ $\mu$s and $t=2000$ $\mu$s , respectively, where at each instance, the response for user $u=1$ and $u=2$ are not only affected by their own delay and other user's spatial interference, but also affected by other users' delayed symbols, which leads to space-time varying joint space-time interference. This dual spatio-temporal variation necessitates joint spatio-temporal precoding using dual 2-dimensional eigenfunctions, which are dual joint space-time orthogonal. 

Figure \ref{fig:Eigenfunctions} shows two pairs of dual spatio-temporal eigenfunctions $(\phi_n(u',t'),\psi_n(u,t))$ (absolute values) 
obtained by decomposing $k_H(u,t;u',t')$ in \eqref{eq:thm2_decomp}.
We see that this decomposition is indeed asymmetric as each $\phi_n(u',t')$ and $\psi_n(u',t')$ are not equivalent, and that each $\phi_1(u',t')$ and $\psi_1(u,t)$ are jointly orthogonal with $\phi_2(u',t')$ and $\psi_2(u,t)$ as in \eqref{eq:thm2_decomp_pty}, respectively.
Therefore, when $\phi_1(u',t')$ (or $\phi_2(u',t')$) is transmitted through the channel, the dual eigenfunctions, $\psi_1(u,t)$ (or $\psi_2(u,t)$) is received with $\sigma_1$ and $\sigma_2$, respectively. 
Therefore, the non-stationary 4-D channel is decomposed to dual flat-fading sub-channels. 
Although Theorem \ref{thm:hogmt} decomposes the channel kernel into infinite 
eigenfunctions, it is sufficient to approximate the channel kernel with a finite number of eigenfunctions with most eigenvalues in terms of mean square errors~\cite{cohen1997nonlinear}. These eigenfunctions are used to calculate the coefficients for joint spatio temporal precoding, which subsequently construct the precoded signal after inverse KLT.

\begin{figure*}[t]
\centering
\begin{subfigure}{.23\textwidth}
  \centering
  \includegraphics[width=1\linewidth]{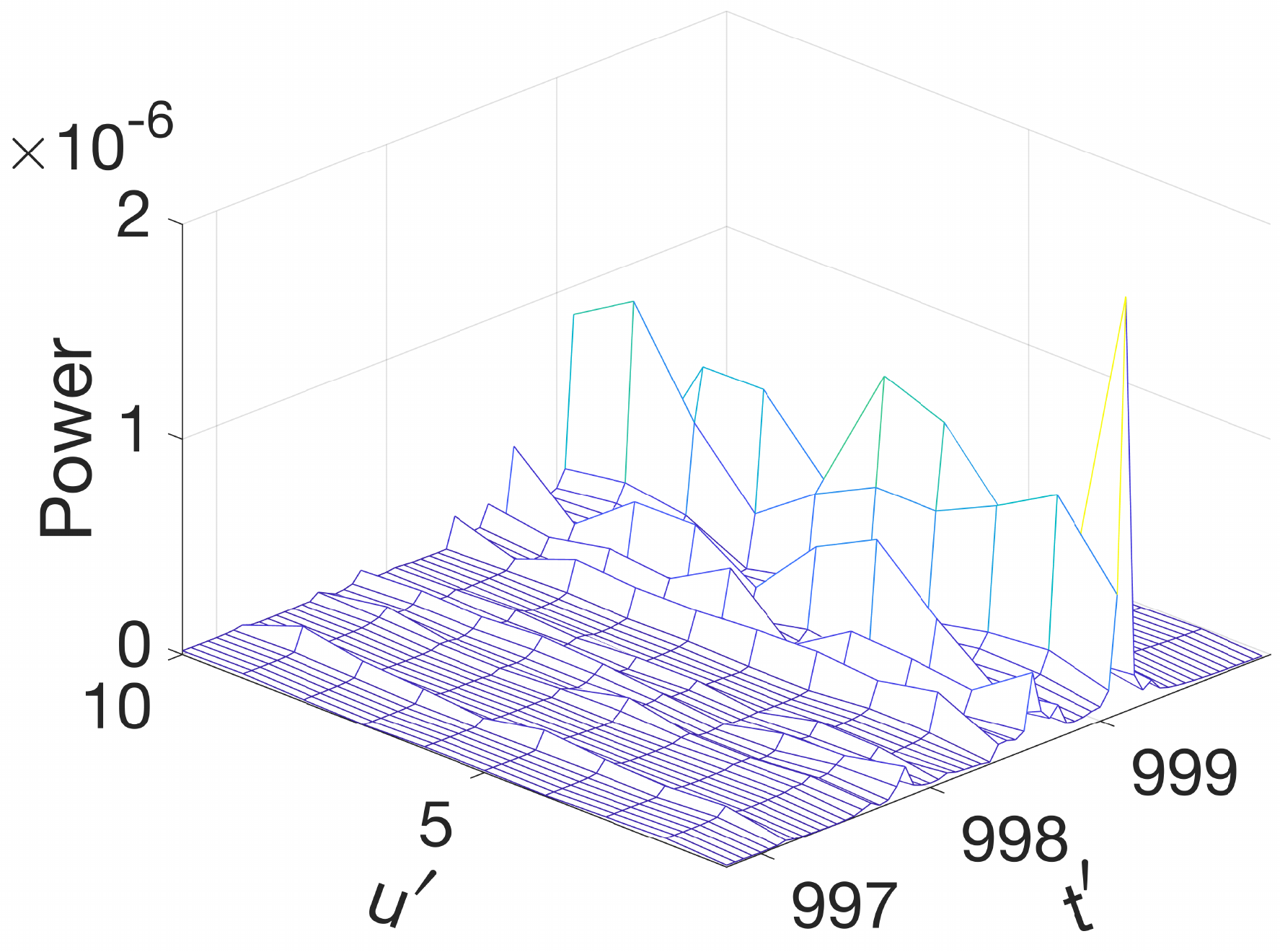}
  \caption{$u = 1$ and $t = 1000$ $\mu$s 
  } 
  \label{fig:hst_1u_1ms}
\end{subfigure}
\begin{subfigure}{.23\textwidth}
  \centering
  \includegraphics[width=1\linewidth]{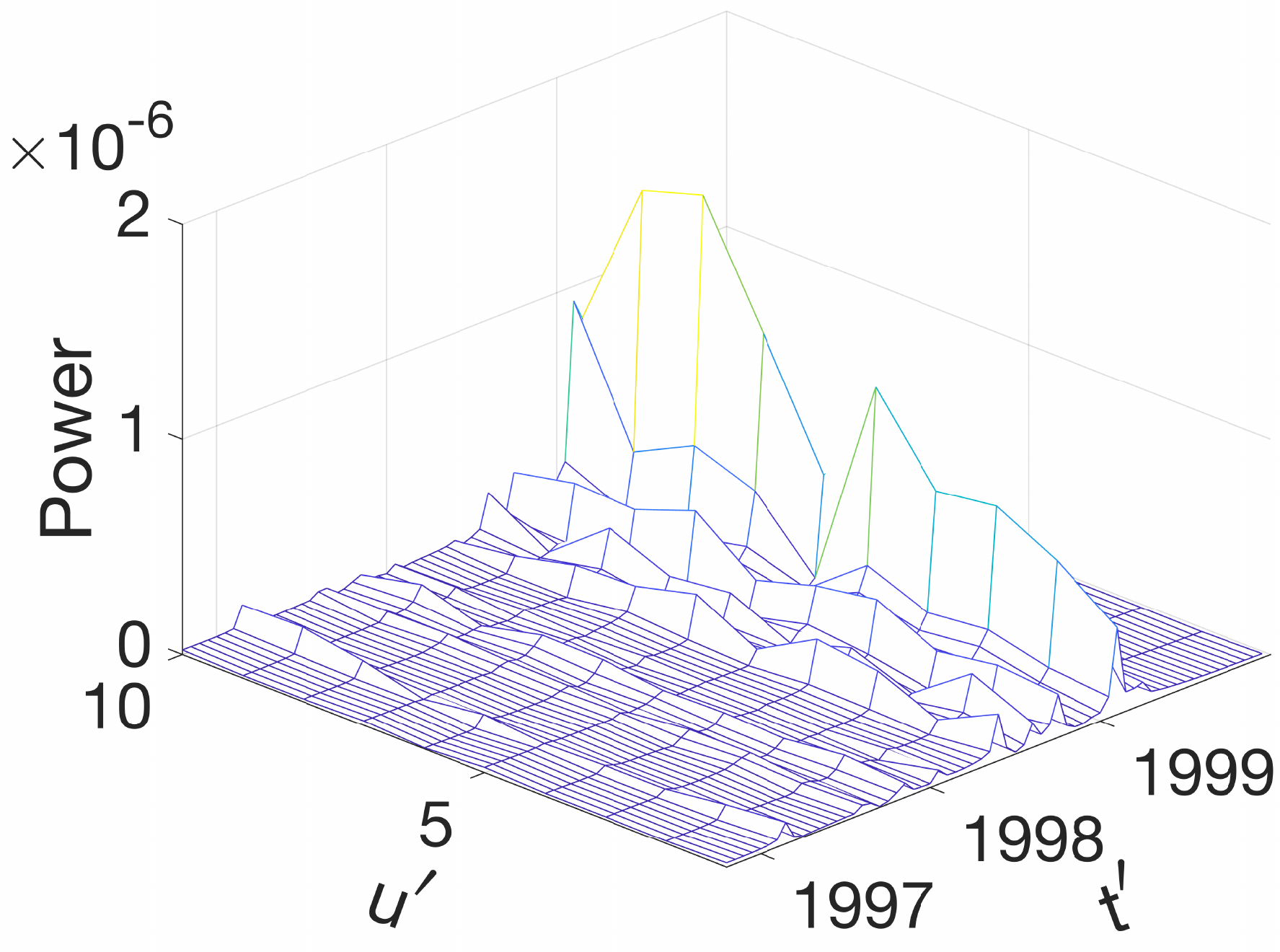}
  \caption{$u = 1$ and $t = 2000$ $\mu$s }
  \label{fig:hst_1u_100ms}
\end{subfigure}
\begin{subfigure}{.23\textwidth}
  \centering
\includegraphics[width=1\linewidth]{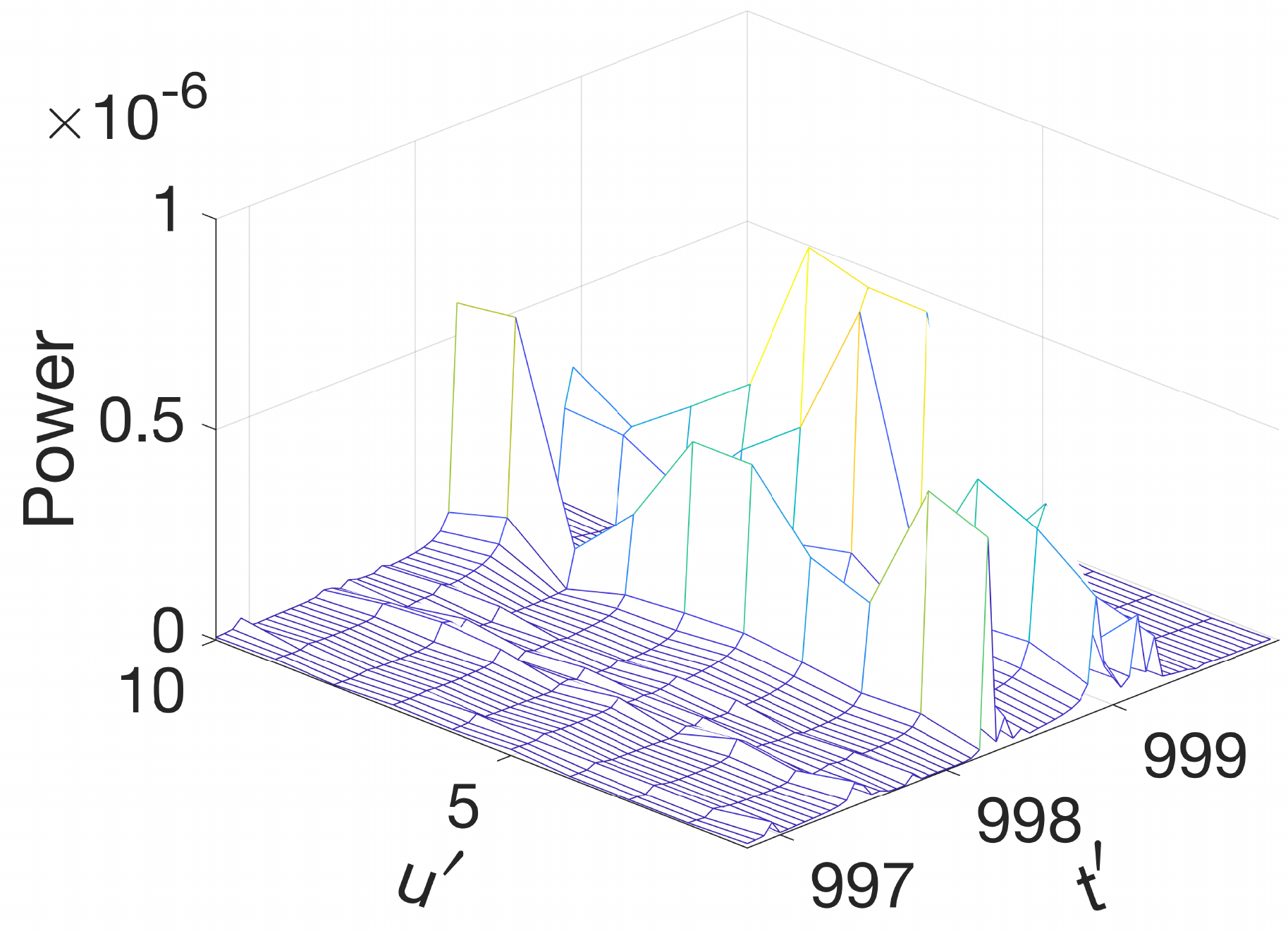}
  \caption{$u = 2$ and $t = 1000$ $\mu$s }
  \label{fig:hst_2u_1ms}
\end{subfigure}
\begin{subfigure}{.23\textwidth}
  \centering
  \includegraphics[width=1\linewidth]{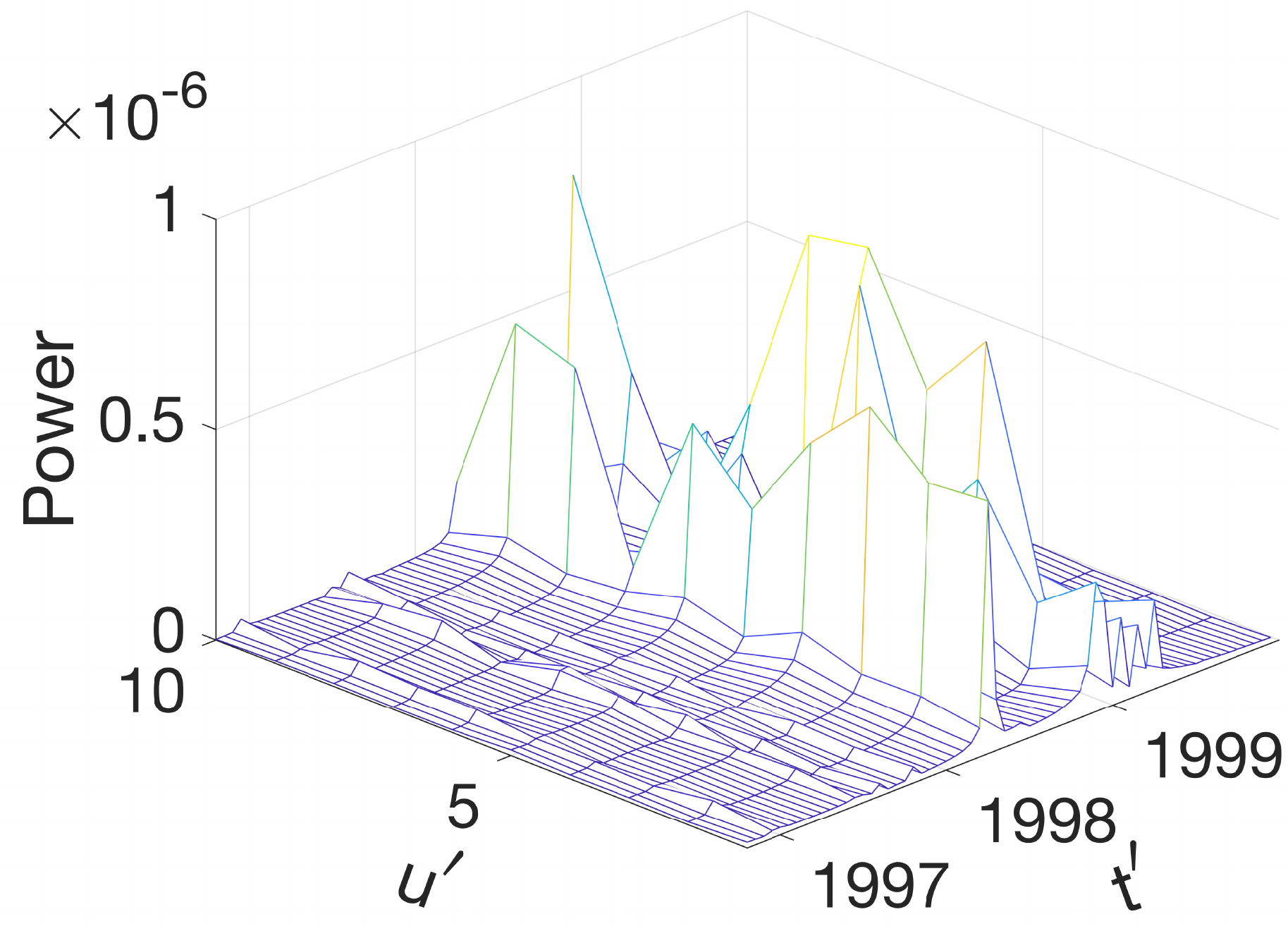}
  \caption{$u = 1$ and $t = 2000$ $\mu$s }
  \label{fig:hst_2u_100ms}
\end{subfigure}
\caption{4-D kernel $k_H(u,t;u',t')$}
\label{fig:4-D kernel}
\end{figure*}

\begin{figure*}[t]
\centering
\begin{subfigure}{.225\textwidth}
  \centering
  \includegraphics[width=1\linewidth]{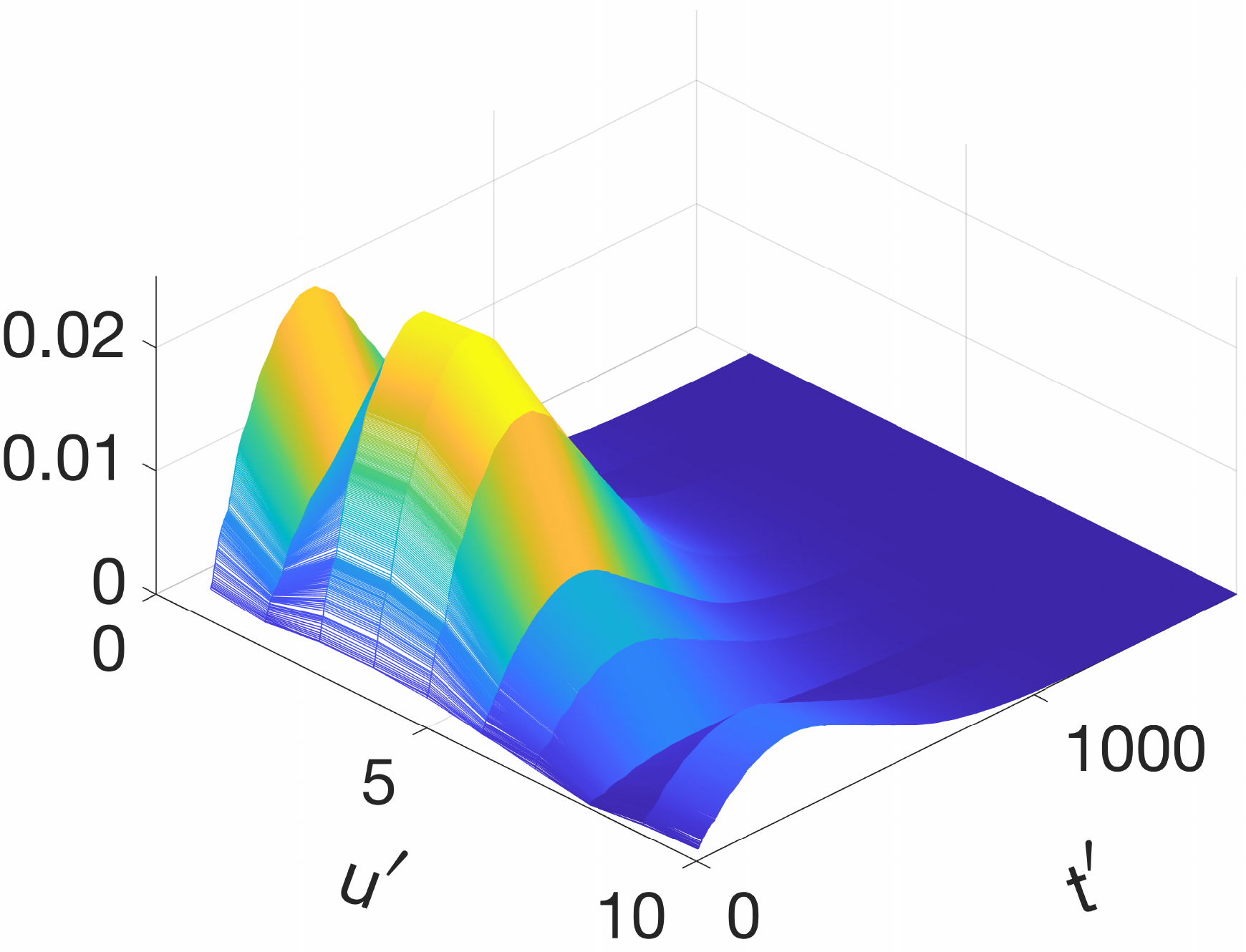}
  \caption{$\phi_1(u^\prime,t^\prime)$
  } 
  \label{fig:phi_1}
\end{subfigure}
\begin{subfigure}{.225\textwidth}
  \centering
  \includegraphics[width=1\linewidth]{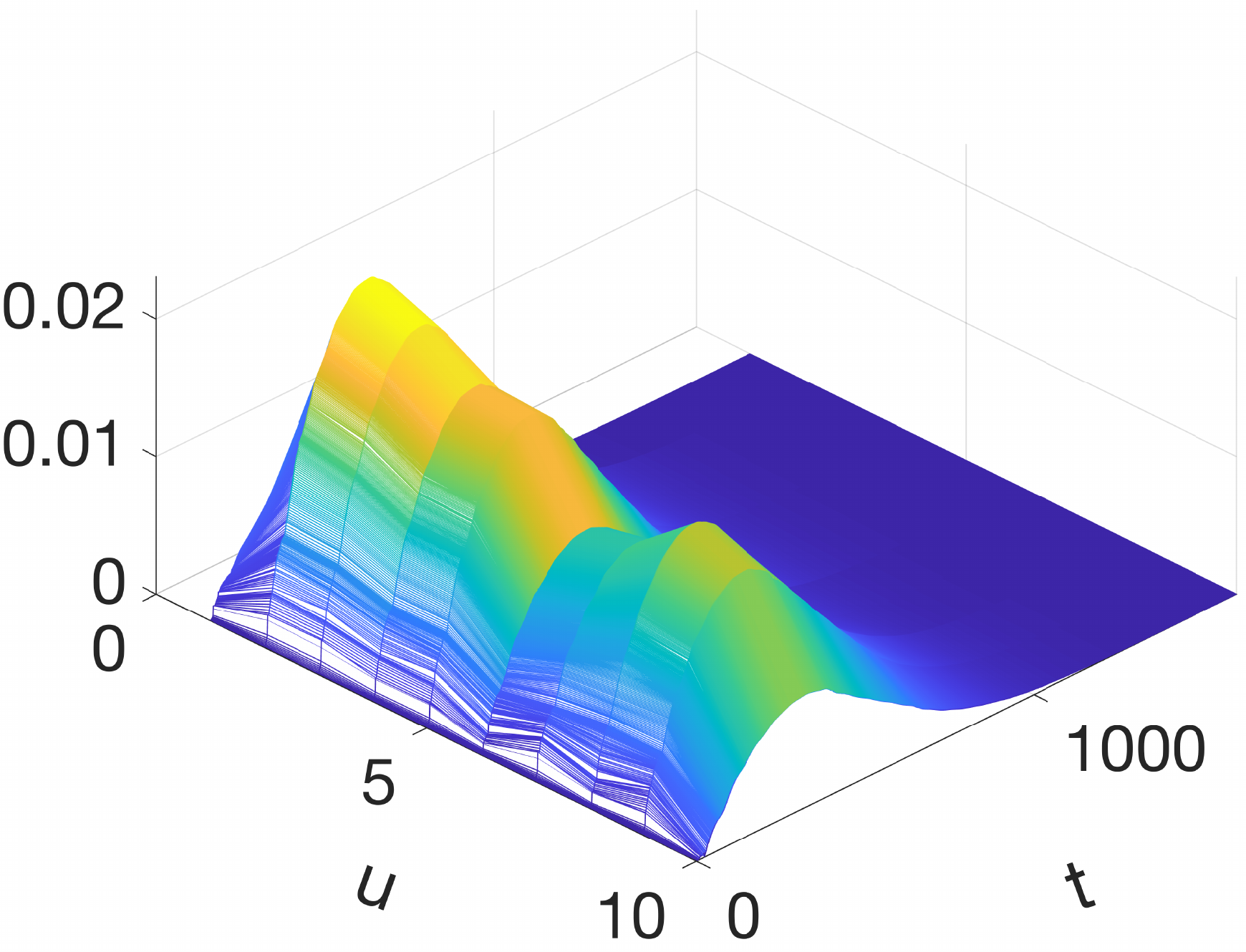}
  \caption{$\psi_1(u,t)$}
  \label{fig:psi_1}
\end{subfigure}
\begin{subfigure}{.225\textwidth}
  \centering
\includegraphics[width=1\linewidth]{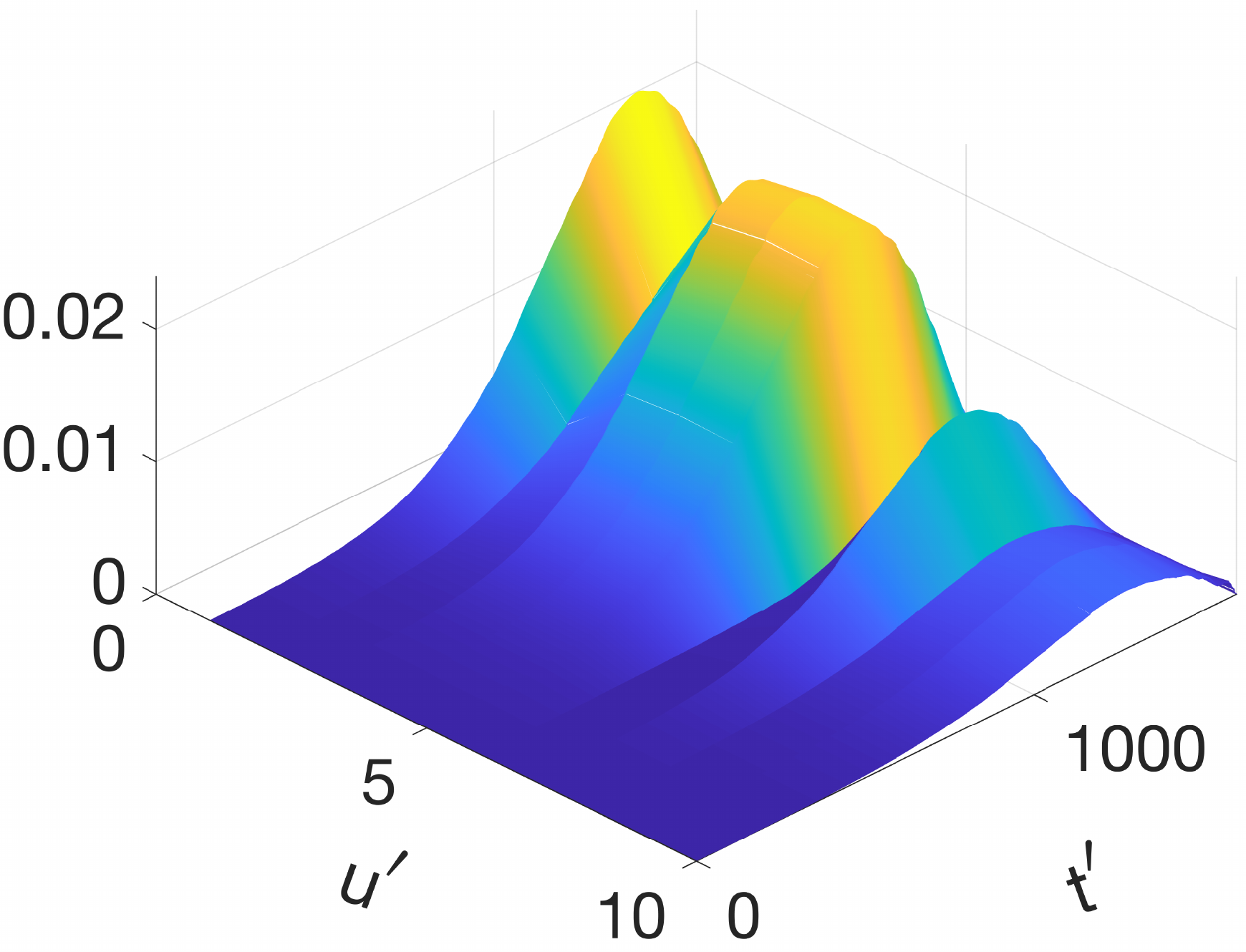}
  \caption{$\phi_2(u^\prime,t^\prime)$}
  \label{fig:phi_2}
\end{subfigure}
\begin{subfigure}{.225\textwidth}
  \centering
  \includegraphics[width=1\linewidth]{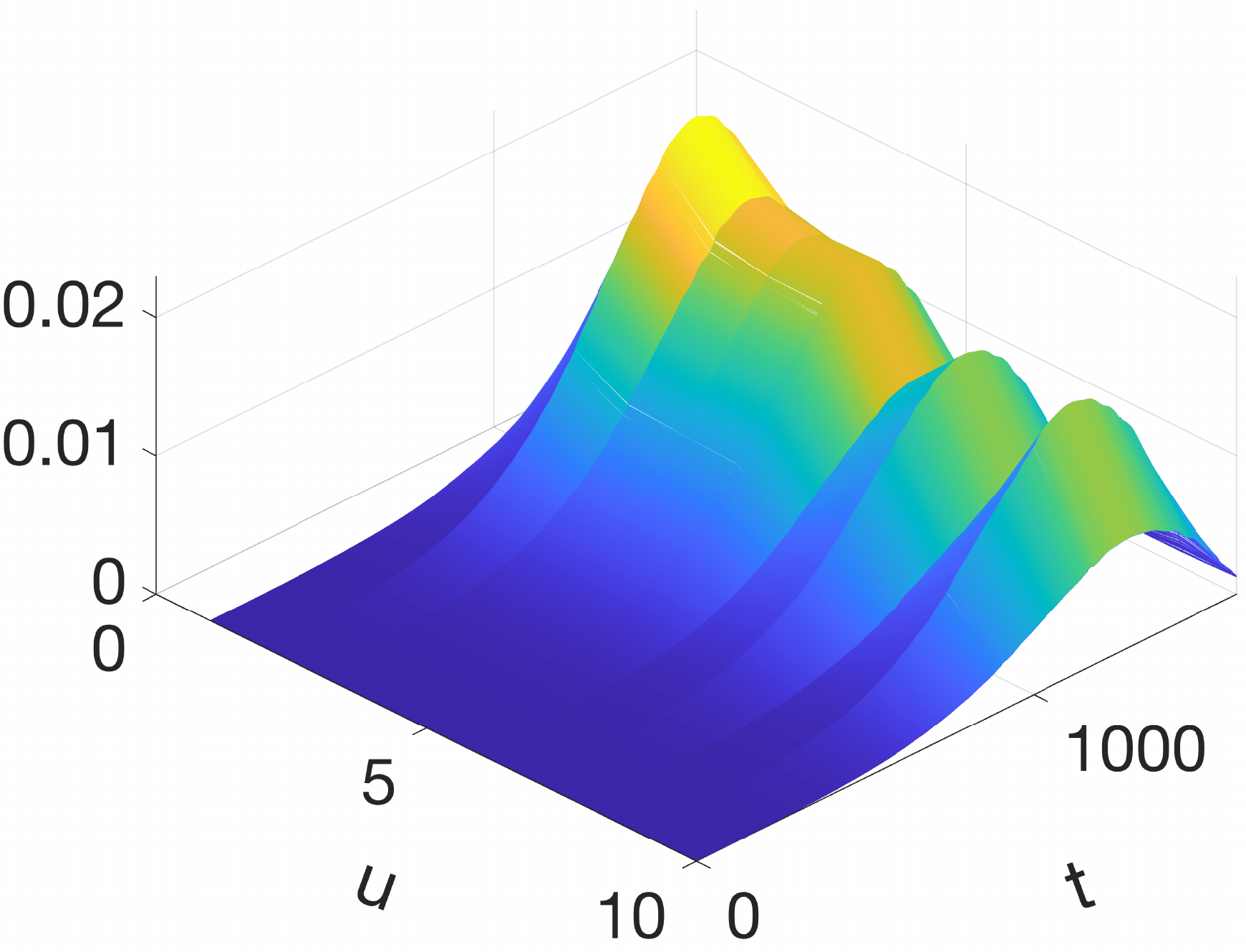}
  \caption{$\psi_2(u,t)$}
  \label{fig:psi_2}
\end{subfigure}
  \caption{Dual spatio-temporal eigenfunctions decomposed from kernel $k_H(u,t;u',t')$}
\label{fig:Eigenfunctions}
\end{figure*}

\begin{figure*}[t]
\centering
\begin{subfigure}{.27\textwidth}
  \centering
    \includegraphics[width=1\linewidth]{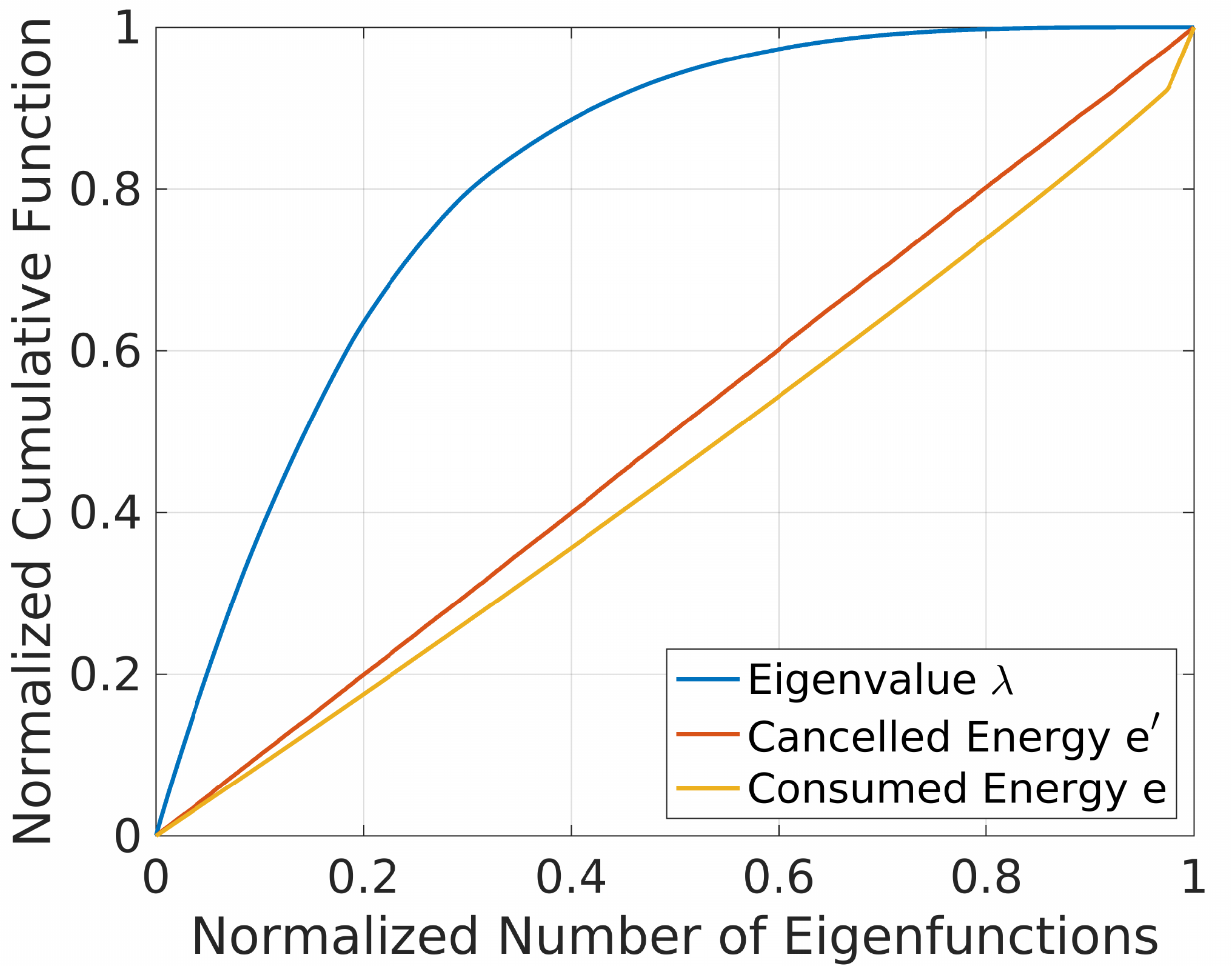}
  \caption{Normalized cumulative function of eigenvalues $\lambda$, cost energy $e_n$ and cancelled interference energy $e_n^\prime$, where both $e_n$ and $e_n^\prime$ are in dB}
  \label{fig:cumulative}
\end{subfigure}
\qquad
\centering
\begin{subfigure}{.27\textwidth}
  \centering
  \includegraphics[width=1\linewidth]{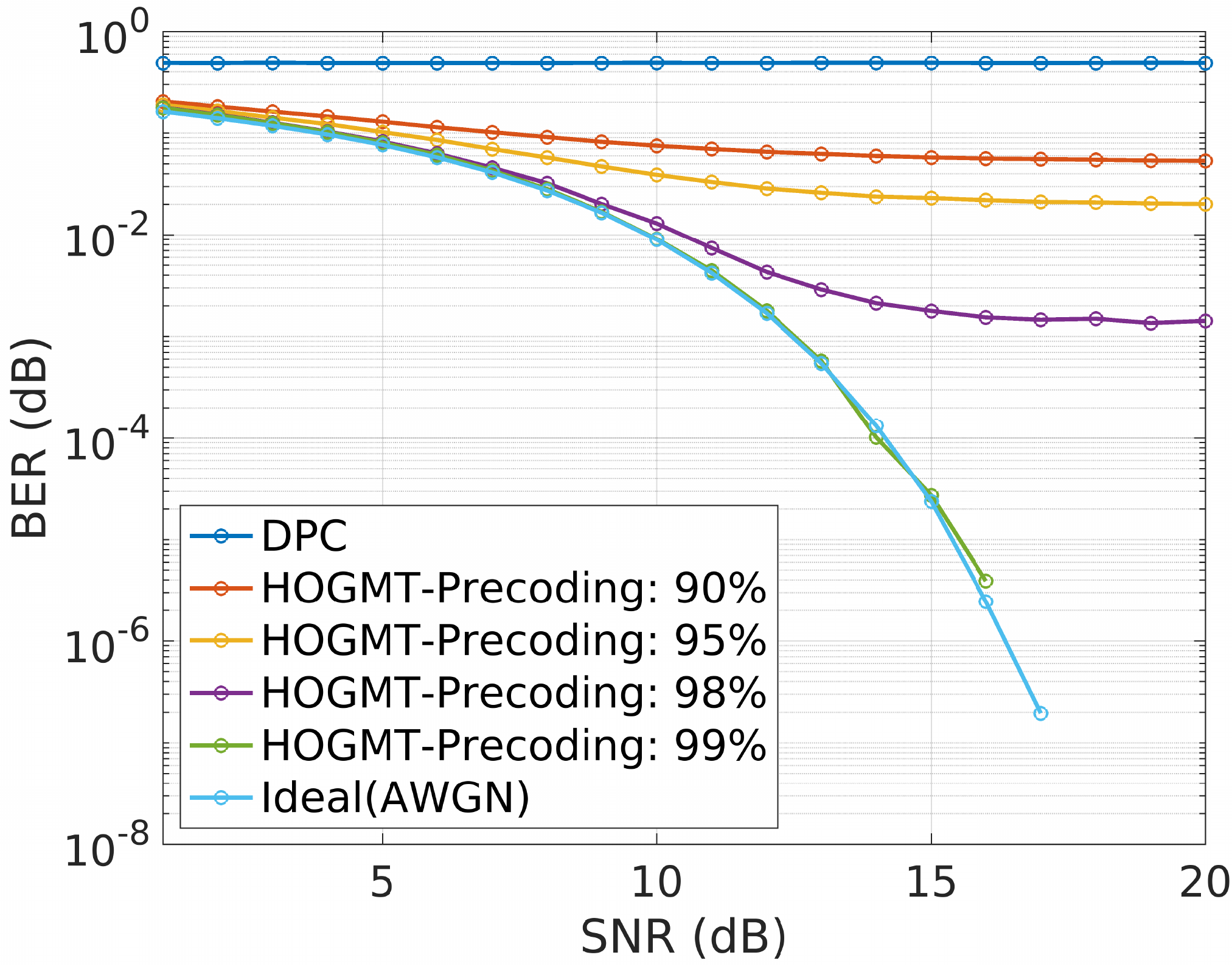}
  \caption{BER of HOGMT based spatial-temporal precoding using different number of eigenfunctions and comparison with the state-of-the-art}
  \label{fig:ber_space_time}
\end{subfigure}
\qquad
\begin{subfigure}{.27\textwidth}
  \centering
  \includegraphics[width=1\linewidth]{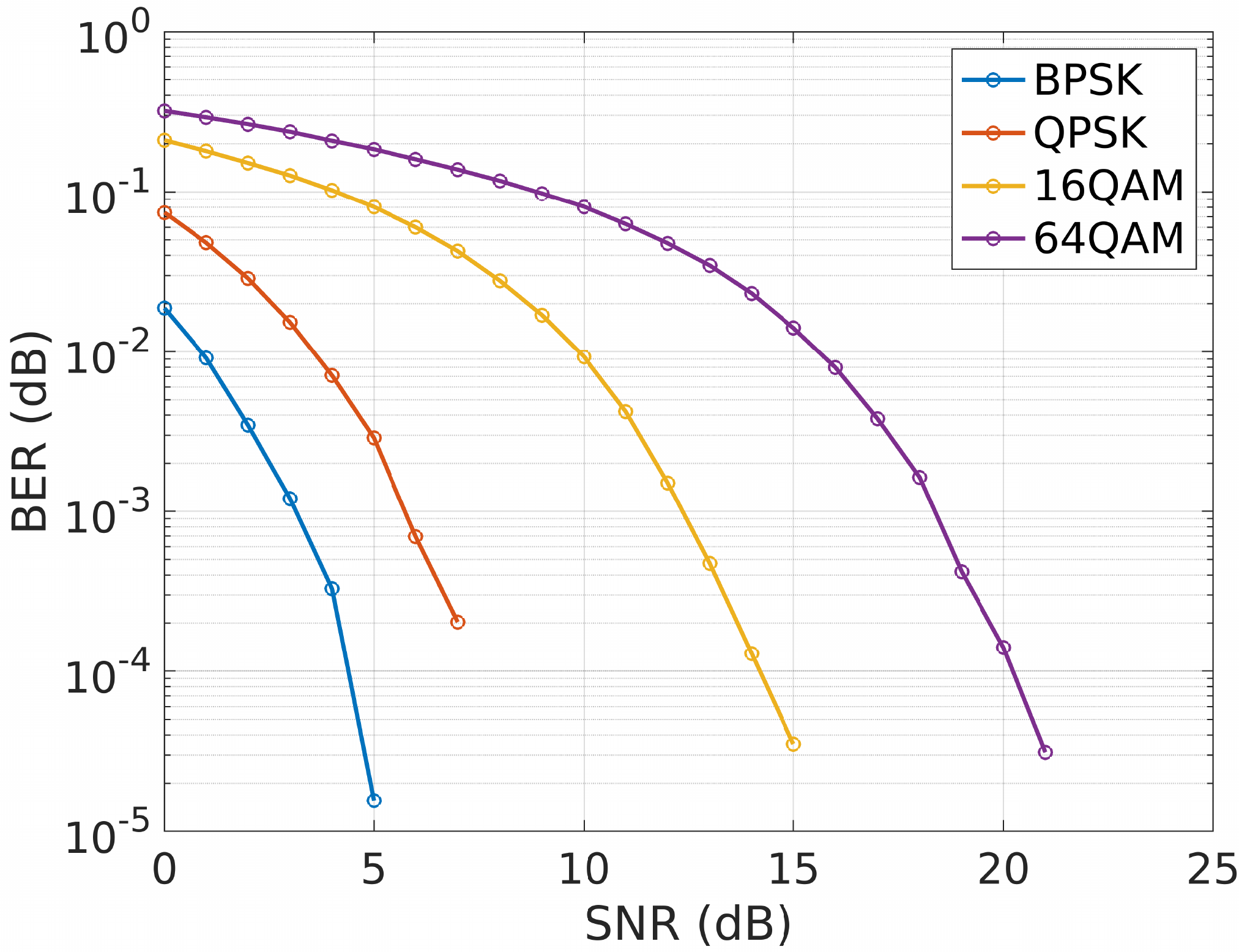}
  \caption{BER of HOGMT based spatio-temporal precoding for BPSK, QPSK, 16-QAM and 64-QAM modulations with $99\%$ eigenfunctions}
  \label{fig:Ber_space_time_multimod}
\end{subfigure}
  \caption{HOGMT based spatio-temporal precoding}
\end{figure*}

In Theorem~\ref{thm:thm2}, the energy allocated to $n$\textsuperscript{th} eigenfunction $\phi_n$ is $e_n = x_n^2$. As the data signal is directly reconstructed by eigenfunctions $\{\psi_n\}$ with $\{s_n\}$, the reconstruction with $n$\textsuperscript{th} eigenfunction $\psi_n$ is equivalent to cancelling an interference $-\psi_n$ with energy $e_n' = s_n^2$. The eigenvalue $\lambda_n$ in \eqref{eq:TFgain} is the transmission gain for $n$\textsuperscript{th} eigenfuntion. From \eqref{eq:x}, we have $e_n = e_n'/\lambda_n$, meaning more energy is required to cancel interference with more energy and less transmission gains.

Figure~\ref{fig:cumulative} shows the normalized cumulative function of eigenvalues $\lambda$, cost energy $e_n$ and cancelled interference energy $e_n'$ with respect to the normalized number of eigenfunctions (with descending order of eigenvalues). We observe that, the interference can be fully cancelled by using all eigenfuntions. However, for the last few eigenfuntions with least eigenvalues, it needs more energy, especially using more than $98\%$ eigenfuntions. Meanwhile, the interference cancellation, i.e., $e_n'$ is basically linear with respect to the number of eigenfuntions used. That's because the linear implementation only decomposing eigenfuntions for approximating the kernel. Eigenfunctions with more eigenvalue extract more information of the channel channel, however, are just orthogonal basis with basically equal contributions for reconstructing the data signal. Theoretically, there exists an optimal nonlinear implementation with respect to maximum energy efficiency, which is beyond the scope of this work.

Figures \ref{fig:ber_space_time} shows the 
BER at the receiver, using joint spatio-temporal precoding (HOGMT-precoding) at the transmitter with 16-QAM modulated symbols for non-stationary channels.
Since this precoding is able to cancel all space-time varying interference that occurs in space, time and across space-time dimensions which are shown in figure \ref{fig:4-D kernel}, 
it achieves significantly lower BER over DPC, which is existing interference-free precoding, however, is applicable for the time-invariant channel matrix thus shows catastrophic performance for non-stationary channels.
Further, we show that with more eigenfuntions, proposed methods achieve lower BER. With more than $99\%$ eigenfunctions, proposed method can achieve near ideal BER, where the ideal case assumes all interference is cancelled and only AWGN noise remains at the receiver.
Figure \ref{fig:Ber_space_time_multimod} compares the BER of HOGMT based spatio-temporal precoding for various modulations (BPSK, QPSK, 16-QAM and 64-QAM) for the same non-stationary channel using $99\%$ eigenfunctions.
As expected we observe that the lower the order of the modulation, the lower the BER but at the cost of lower data rate. 
However, we observe that even with high-order modulations (\eg 64-QAM) the proposed precoding achieves low BER (${\approx}10^{-4}$ at SNR${=}20$dB), allowing high data-rates even over challenging non-stationary channels. The choice of the order of the modulation is therefore, based on the desired BER and data rate for different non-stationary scenarios.
\section{Conclusion}
\label{sec:conclusion}

In this work, we derived a high-order generalized version of Mercer's Theorem to decompose the high-order asymmetric kernels into dual 2-dimensional jointly orthogonal eigenfuntions. Through theoretical analysis and simulations, we draw three firm conclusions for non-stationary channels: 1) the 2-dimensional eigenfunctions decomposed from 4-dimensional coefficients of atomic channels across time-frequency and delay-Doppler domain are sufficient to completely derive the second-order statistics of the non-stationary channel and consequently leads to an 
unified characterization of any wireless channel, 2) The duality and joint orthogonality of 2-dimensional eigenfunctions decomposed from 4-dimensional non-stationary channels manifest independently flat-fading, 3) precoding by these eigenfunctions with optimally derived coefficients 
mitigates the spatio-temporal interference 
, 4) the precoded symbols when propagated over the non-stationary channel directly reconstruct the modulated symbols at the receiver when combined with the calculated coefficients, consequently alleviating the need for complex complementary step at the receiver and
5) proposed precoding has less complexity than DPC.
Therefore, the encouraging results from this work will form the core of robust and unifed characterization and highly reliable communication over non-stationary channels, supporting emerging application.

\section{ACKNOWLEDGEMENT}
This work is funded by the Air Force Research Laboratory Visiting Faculty Research Program (SA10032021050367), Rome, New York, USA.

\bibliographystyle{IEEEtran}
\bibliography{references}



\end{document}